%% file: well-structured-jgts.tex
\documentclass[submission]{eptcs}

\usepackage{./sources/wsgts}

\providecommand{\event}{GCM 2021} 

\title{Resilience of Well-structured Graph Transformation Systems\thanks{
Supported by the German Research Foundation (DFG) through the
Research Training Group (DFG GRK 1765) SCARE}}
\author{Okan \"Ozkan \qquad\qquad\qquad Nick W\"urdemann
\institute{Department of Computing Science\\
	University of Oldenburg\\
	Oldenburg, Germany}
\email{\{o.oezkan,wuerdemann\}@informatik.uni-oldenburg.de}
}
\def\titlerunning{Resilience of Well-structured GTSs}
\def\authorrunning{O. \"Ozkan \& N. W\"urdemann}
\begin{document}
\maketitle

\begin{abstract}
Resilience is a concept of rising interest in computer science and software engineering.
For systems in which correctness w.r.t.\ a safety condition is unachievable, fast recovery is demanded.
We investigate resilience problems of graph transformation systems.
Our main contribution is the decidability of two resilience problems for well-structured graph transformation systems (with strong compatibility).
We prove our results in the abstract framework of well-structured transition systems and apply them to graph transformation systems,
incorporating also the concept of adverse conditions.
  
\end{abstract}

\section{Introduction}
\input{./sources/introduction}
\section{Preliminaries}
\input{./sources/preliminaries}

\input{./sources/contribution}

\input{./sources/application}

\section{Related Work}
\input{./sources/relatedwork}
\section{Conclusion}
\input{./sources/conclusion}

\bibliographystyle{eptcs}
\bibliography{./sources/bibWithDOIs} 
\appendix
\end{document}

%% file: sources/introduction.tex
\emph{Resilience} is a broadly used concept in computer science and software engineering (e.g.,\ \cite{Trivedi09}), 
and a basic concept for, e.g.,\ industrial control systems \cite{Rieger13} and mobile cyber-physical systems \cite{Karsai16}.
For systems in which \emph{correctness} w.r.t.\ a safety condition $\SAFE$ is unachievable, 
\emph{fast recovery} is demanded. We interprete fast recovery as reachability of the
safety condition in a bounded amount of time steps. The intuitive approach is 
to start from any \emph{error state}, i.e., a state in which $\neg\SAFE(\equiv\ERR)$ holds, and try
to reach a state in which $\SAFE$ holds again as fast as possible. 

Another approach to formalizing resilience is to ask whether 
the system can \emph{withstand an adverse effect} rather than to ask
whether fast recovery is possible from any error state. To formally capture adverse 
effects we consider an environment interacting with the system.
In this setting, we investigate on the question whether a state satisfying $\SAFE$ can be reached in bounded time, 
starting from any state satisfying $\ENV$, i.e.,
any state directly resulting from an environment interference.

For modeling systems we use \emph{graph transformation systems (GTSs)}, 
as considered, e.g., in \cite{Ehrig97}, which are a visual yet precise formalism. 
In this perception, system states are captured by graphs and state changes by graph
transformations. Usually, the state set (the set of graphs reachable from a start graph) is infinite.
To handle infinite state sets, 
we incorporate the concept of well-structuredness \cite{AbdullaCJT96,FinkelS01, Koenig17}. 
A \emph{well-structured transition system (WSTS)} is informally 
a transition system equipped with a well-quasi-order (wqo) satisfying that
larger states simulate smaller states. 
This allows us to abstract from both of the approaches towards resilience described above. 
In the setting of WSTSs, we define resilience problems for a given downward-closed set $\Ji$ 
(a $ \envi $ condition, e.g.,\ $\ERR$ or $\ENV$) 
and an upward-closed set~$\Ii$ (e.g.,\ a safety property $\SAFE$).
Given an initial state~$ s $ and a natural number $k$, 
the \emph{explicit resilience problem} asks
whether we can, starting from $s$, reach $\Ii$ in at most $k$ steps 
whenever we reach $\Ji$. 
The \emph{bounded resilience problem}
asks whether there exists a~$ k $ such that 
\emph{k-step resilience} is satisfied.

We show that both resilience problems (given a basis of the upward-closure of the reachable states) are 
decidable for \emph{strongly} well-structured transition systems (SWSTSs).
We propose an algorithm which computes the minimal $k$ s.t.\ we can 
recover from any $ \envi $ state in at most $k$ steps, 
or returns $\false$ if there exists no such $k$. 
It is based on the ideal reachability algorithm proposed by Abdulla et al.\ \cite{AbdullaCJT96},
and solves both resilience problems at the same time. 

When applying these results to GTSs, 
we assume that the corresponding graph class is of bounded path length 
in order to obtain a SWSTS. 
This sufficient condition for a GTS to be strongly well-structured is shown by K{\"o}nig \& St{\"u}ckrath in \cite{Koenig17}. 
The wqo on graphs used in this case is the subgraph order, so $\Ii =\Ii_\SAFE$ corresponds to a constraint stating existence of subgraphs.
We incorporate adverse conditions by distinguishing system and environment rules,
and considering $\Ji= J_\ENV$, the set of graphs directly resulting from the application of an \mbox{environment}~rule. 

The rest of this paper is organized as follows: 
We recall preliminary concepts in Sec.~\ref{sec:prelim}. 
In Sec.~\ref{sec:adverseCondAndResProblems}, we present the concept of resilience in the context of adverse conditions
and identify abstract resilience problems. 
In Sec.~\ref{sec:decidabilityResults}, we prove decidability of resilience for strongly well-structured transition systems. 
We apply these results to graph transformation systems incorporating adverse conditions in Sec.~\ref{sec:appl}. 
In Sec.~\ref{sec:related}, we present related work. 
We close with a conclusion and an outlook in Sec.~\ref{sec:conc}. 

%% file: sources/preliminaries.tex
\label{sec:prelim}
We recall the concepts used in this paper, namely
\emph{graph transformation systems} \cite{Ehrig97,Ehrig06} and 
(in particular \emph{well-structured}) \emph{transition systems} \cite{FinkelS01}.

\subsection{Graph Transformation Systems}
In the following, we recall the definitions of graphs, graph conditions, 
rules, and graph transformation systems \cite{Ehrig97,Ehrig06}. 
A directed, labeled graph consists of a set of nodes 
and a set of edges where each edge is equipped with 
a source and a target node and where each node and 
edge is equipped with a label. Note that this kind of graphs
are a special case of the hypergraphs considered in \cite{Koenig17}.

\begin{defn}[graphs \& graph morphisms]  
	A \emph{(directed, labeled) graph} (over a finite label alphabet~$\Lambda$)
	is a tuple $\G=\tuple{\V_\G,\E_\G,\sou_\G,\tar_\G,\labV_\G,\labE_\G}$,
	with finite sets $\V_\G$ and $\E_\G$ of \emph{nodes} (or \emph{vertices})
	and \emph{edges}, 
	functions $\sou_\G,\tar_\G:\E_\G\to \V_\G$ 
	assigning \emph{source} and \emph{target} to each edge, and 
	\emph{labeling functions} $\labV_\G:\V_\G\to\Lambda$, $\labE_\G:\E_\G\to\Lambda$ . 
	A \emph{(simple, undirected) path} $p$ in $\G$ of length $\ell$ is a sequence 
	$\tuple{v_1 , e_1 , v_2 \ldots, v_\ell , e_\ell,  v_{\ell+1}}$ of nodes and edges s.t.\ 
	$\sou_\G(e_i )=v_i$ and $\tar_\G(e_i )=v_{i+1}$, or $\tar_\G(e_i )=v_i$ and $\sou_\G(e_i )=v_{i+1}$ for every $1 \le i \le \ell$, 
	and all contained nodes and edges occur at most once.	
	Let $\ell(\G)$ denote the length of a longest path in $G$. 
	Given graphs $\G$ and $\GH$, a \emph{(partial graph) morphism} 
	$\mor: \G \pmor \GH$ consists of partial functions 
	$\mor_\V:\V_\G\pmor\V_\GH$ and $\mor_\E:\E_\G\pmor\E_\GH$ 
	which preserve sources, targets, and labels, 
	i.e., $\mor_\V\circ\sou_\G (e)=\sou_\GH\circ \mor_\E (e)$, 
	$\mor_\V\circ\tar_\G(e)=\tar_\GH\circ \mor_\E(e)$, 
	$\labV_{\G} (v)=\labV_{\GH}\circ \mor_\V (v)$, and
	$\labE_{\G} (e)=\labE_{\GH}\circ \mor_\E (e)$ on all egdes $e$ and nodes $v$, for which
	$\mor_E(e), \labE(e), \labV(v)$ is defined. Furthermore, if a morphism is defined on an edge, it must be defined on all incident nodes. 
	The morphism $g$ is \emph{total} (\emph{injective}) if both $\mor_{\V}$ and $\mor_{\E}$ 
	are total (injective). 
	If $\mor$ is total and injective, we also write $\mor: \G \injto \GH$. 
	The composition of morphisms is defined componentwise.
\end{defn}

We consider graph constraints \cite{Rensink04,HabelP09} whose validities are inherited to 
bigger/smaller~graphs.
\begin{defn}[positive \& negative basic graph constraints] 
	The class of \emph{positive (basic graph) constraints} is defined inductively: 
	(i) $\exists G$ is a positive constraint where 
	$G$ is a graph, 
	(ii) for positive constraints $c,c'$, also 
	$c \lor c'$, $ c\land c'$ are positive constraints. 
	Analogously, the \emph{negative (basic graph) constraints} are defined by: 
	(i) $\neg \exists G$ is a negative constraint for any graph $G$, 
	(ii) for negative constraints $c,c'$, also
	$c \lor c'$, $ c\land c'$ are negative constraints. 
A graph $G$ \emph{satisfies} $\exists G'$ if there exists an total injective morphism $G' \injto G$. The semantics of the logical operators are as usual.
We write $G \models c$ if $G$ satisfies the positive/negative constraint $c$.
\end{defn}
\begin{rem}
	If $c$ is a positive constraint, $\neg c$ is equivalent to a negative constraint, and vice versa.
\end{rem}

\begin{fact}[upward \& downward inheritance] 
	Let $G \injto H$ be a total injective morphism, $c$ be a positive constraint, 
	and $c'$ a negative constraint. 
	If $G \models c$, then also $H \models c$. 
	If $H \models c'$, then also $G \models c'$.
\end{fact}

We use the \emph{single pushout (SPO)} approach \cite{Ehrig97,Koenig17} 
with injective matches for modeling graph transformations. 
The reason for choosing SPO and not, e.g.,\ the \emph{double pushout approach (DPO)} \cite{Ehrig06} is that the dangling 
condition disturbs the compatibility condition of WSTS in Def.~\ref{wellstr}. 

\begin{defn}[rules \& transformations]\label{def:rulesAndTransformations}
	A \emph{(graph transformation) rule} $r=\tuple{L \pmor R}$ 
	(over a finite label alphabet $\Lambda$) is a partial morphism from $L$ to $R$ (both graphs over $\Lambda$).
	A \emph{(direct) transformation} $G \trafo H$ from a graph $G$ to a 
	graph $H$ applying rule $r$ at a total injective \emph{match morphism} $g: L \injto G$ is given by
	a \emph{pushout} as shown in Fig. \ref{fig:GraphTransformationVisualization} (for existence and construction 
	of pushouts, see, e.g., \cite{Ehrig97}). We write $G \trafo_r H$ to indicate the applied rule, and
	$G \trafo_\R H$ if $G \trafo_r$ for a rule $r$ contained in the rule set $\R$.
\end{defn}
Note that we do not have any application conditions. 
The pushout of a rule application
is visualized in Fig.~\ref{fig:GraphTransformationVisualization}.
An example for a rule is presented in Fig.~\ref{fig:GTRule},
and an application of that rule in Fig.~\ref{fig:ruleApplication}.
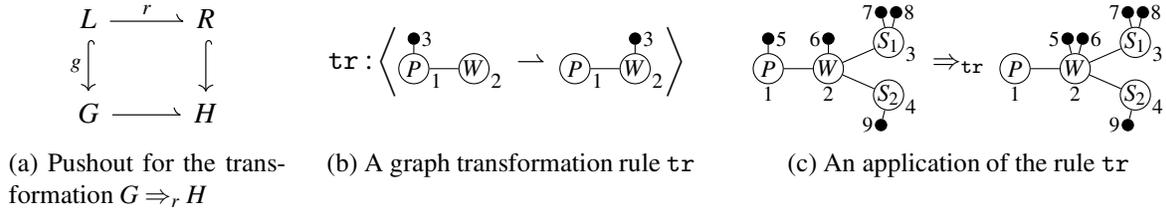
\begin{figure}[!htb]
	\centering
	\begin{subfigure}[b]{0.23\textwidth}
		\centering
		\begin{tikzcd}[>=stealth']
		L \arrow[r, "r",harpoon] \arrow[d, "g"', hook] & R \arrow[d, hook] \\
		G \arrow[r, harpoon]                       & H                     
		\end{tikzcd}
		\caption{Pushout for the transformation $G \dder_r H$}
		\label{fig:GraphTransformationVisualization}
	\end{subfigure}
	\hspace*{0.015\textwidth}
	\begin{subfigure}[b]{0.32\textwidth}
		\centering
		$ \mathtt{tr:} 
		\left\langle
		\begin{tikzpicture}[baseline=(product.center)]
		\renewcommand{\xdis}{8mm}
			\node[circle, draw=black, inner sep=0pt, minimum size=4mm] (product) [label={[inner sep=0.5pt]-10:{\scriptsize $ 1 $}}] {{\footnotesize $ P $}};
			\node[circle, draw=black, inner sep=0pt, minimum size=4mm, right=of product] (warehouse) [label={[inner sep=0.5pt]-10:{\scriptsize $ 2 $}}] {{\footnotesize $ W $}};
			\node[circle, inner sep=0pt, draw=black, fill=black, minimum size=1.5mm] at (product) [yshift=4mm] (tokenP) [label={[inner sep=0.5pt]right:{\scriptsize $ 3 $}}] {};
			\draw[-]
			(product) edge (warehouse)
			(product) edge (tokenP)
			;
		\end{tikzpicture}
		~\rightharpoonup~ 
		\begin{tikzpicture}[baseline=(product.center)]
		\renewcommand{\xdis}{8mm}
			\node[circle, draw=black, inner sep=0pt, minimum size=4mm] (product) [label={[inner sep=0.5pt]-10:{\scriptsize $ 1 $}}]{{\footnotesize $ P $}};
			\node[circle, draw=black, inner sep=0pt, minimum size=4mm, right=of product] (warehouse) [label={[inner sep=0.5pt]-10:{\scriptsize $ 2 $}}] {{\footnotesize $ W $}};
			\node[circle, inner sep=0pt, draw=black, fill=black, minimum size=1.5mm] at (warehouse) [yshift=4mm] (tokenW) [label={[inner sep=0.5pt]right:{\scriptsize $ 3 $}}] {};
			\draw[-]
			(product) edge (warehouse)
			(warehouse) edge (tokenW)
			;
		\end{tikzpicture} 
		\right\rangle$
		\vspace*{5mm}
		\caption{A graph transformation rule $ \mathtt{tr} $\\~}
		\label{fig:GTRule}
	\end{subfigure}
	\hspace*{0.015\textwidth}
	\begin{subfigure}[b]{0.37\textwidth}
		\centering
		\begin{tikzpicture}[baseline=(product.center)]
		\renewcommand{\xdis}{8mm}
		\renewcommand{\ydis}{3.5mm}
		\node[circle, draw=black, inner sep=0pt, minimum size=4mm] (product)  [label={[inner sep=0.5pt]-90:{\scriptsize $ 1 $}}] {{\footnotesize $ P $}};
		\node[circle, draw=black, inner sep=0pt, minimum size=4mm, right=of product] (warehouse)  [label={[inner sep=0.5pt]-90:{\scriptsize $ 2 $}}] {{\footnotesize $ W $}};
		\node[circle, draw=black, inner sep=0pt, minimum size=4mm, above right=of warehouse] (store1) [label={[inner sep=0.5pt]-10:{\scriptsize $ 3 $}}] {{\footnotesize $ S_1 $}};
		\node[circle, draw=black, inner sep=0pt, minimum size=4mm, below right=of warehouse] (store2) [label={[inner sep=0.5pt]-10:{\scriptsize $ 4 $}}]{{\footnotesize $ S_2 $}};
		\node[circle, inner sep=0pt, draw=black, fill=black, minimum size=1.5mm] at (product) [yshift=4mm] (tokenP) [label={[inner sep=0.5pt]right:{\scriptsize $ 5 $}}] {};
		\node[circle, inner sep=0pt, draw=black, fill=black, minimum size=1.5mm] at (warehouse) [yshift=4mm] (tokenW) [label={[inner sep=0.5pt]left:{\scriptsize $ 6 $}}] {};
		\node[circle, inner sep=0pt, draw=black, fill=black, minimum size=1.5mm, xshift=-1mm] at (store1) [yshift=4mm] (tokenS11) [label={[inner sep=0.5pt]left:{\scriptsize $ 7 $}}] {};
		\node[circle, inner sep=0pt, draw=black, fill=black, minimum size=1.5mm, xshift=1mm] at (store1) [yshift=4mm] (tokenS12) [label={[inner sep=0.5pt]right:{\scriptsize $ 8 $}}] {};
		\node[circle, inner sep=0pt, draw=black, fill=black, minimum size=1.5mm, xshift=-1mm] at (store2) [yshift=-4mm] (tokenS21) [label={[inner sep=0.5pt]left:{\scriptsize $ 9 $}}] {};
		\draw[-]
		(product) edge (warehouse)
		(warehouse)	edge (store1)
		(warehouse)	edge (store2)
		(product) edge (tokenP)
		(warehouse) edge (tokenW)
		(store1) edge (tokenS11)
				edge (tokenS12)
		(store2) edge (tokenS21)
		;
		\end{tikzpicture}
		~$ \Rightarrow_\mathtt{tr} $~
		\begin{tikzpicture}[baseline=(product.center)]
		\renewcommand{\xdis}{8mm}
		\renewcommand{\ydis}{3.5mm}
		\node[circle, draw=black, inner sep=0pt, minimum size=4mm] (product)[label={[inner sep=0.5pt]-90:{\scriptsize $ 1 $}}] {{\footnotesize $ P $}};
		\node[circle, draw=black, inner sep=0pt, minimum size=4mm, right=of product] (warehouse)[label={[inner sep=0.5pt]-90:{\scriptsize $ 2 $}}] {{\footnotesize $ W $}};
		\node[circle, draw=black, inner sep=0pt, minimum size=4mm, above right=of warehouse] (store1) [label={[inner sep=0.5pt]-10:{\scriptsize $ 3 $}}] {{\footnotesize $ S_1 $}};
		\node[circle, draw=black, inner sep=0pt, minimum size=4mm, below right=of warehouse] (store2) [label={[inner sep=0.5pt]-10:{\scriptsize $ 4 $}}] {{\footnotesize $ S_2 $}};
		\node[circle, inner sep=0pt, draw=black, fill=black, minimum size=1.5mm, xshift=-1mm] at (warehouse) [yshift=4mm] (tokenW1) [label={[inner sep=0.5pt]left:{\scriptsize $ 5 $}}] {};
		\node[circle, inner sep=0pt, draw=black, fill=black, minimum size=1.5mm, xshift=1mm] at (warehouse) [yshift=4mm] (tokenW2) [label={[inner sep=0.5pt]right:{\scriptsize $ 6 $}}] {};
		\node[circle, inner sep=0pt, draw=black, fill=black, minimum size=1.5mm, xshift=-1mm] at (store1) [yshift=4mm] (tokenS11)[label={[inner sep=0.5pt]left:{\scriptsize $ 7 $}}] {};
		\node[circle, inner sep=0pt, draw=black, fill=black, minimum size=1.5mm, xshift=1mm] at (store1) [yshift=4mm] (tokenS12) [label={[inner sep=0.5pt]right:{\scriptsize $ 8 $}}]{};
		\node[circle, inner sep=0pt, draw=black, fill=black, minimum size=1.5mm, xshift=-1mm] at (store2) [yshift=-4mm] (tokenS21) [label={[inner sep=0.5pt]left:{\scriptsize $ 9 $}}]{};
		\draw[-]
		(product) edge (warehouse)
		(warehouse)	edge (store1)
		(warehouse)	edge (store2)
		(warehouse) edge (tokenW1)
		(warehouse) edge (tokenW2)
		(store1) edge (tokenS11)
				edge (tokenS12)
		(store2) edge (tokenS21)
		;
		\end{tikzpicture} 
		\caption{An application of the rule $ \mathtt{tr} $\\~}
		\label{fig:ruleApplication}
	\end{subfigure}
	\caption{Pushout and example of a direct graph transformation.}
	\label{fig:GraphTransformation}
\end{figure}

GTSs are simply finite sets of rules. We specify the state set later. 

\begin{defn}[graph transformation system]
A \emph{graph transformation system (GTS)} is a finite set of 
graph transformation rules.
\end{defn}

\subsection{Transition Systems}
We recall the notion of transition systems. 
In Sec.~\ref{sec:decidabilityResults}, we prove our results on the level of transition systems 
and explicate the concept for graph transformation systems in Sec.~\ref{sec:appl}.
\begin{defn}[transition system]
	A \emph{transition system (TS)} $\tuple{\St,\pfeil}$ 
	consists of a (possibly infinite) set $\St$ of \emph{states} 
	and a \emph{transition relation} $\pfeil \subseteq \St \times \St$. 
	Let $\pfeil^0=\Id_\St$ (identitiy on $\St$), $\pfeil^1=\pfeil$, and 
	$\pfeil^k = \pfeil^{k-1} \circ \pfeil$ for every $k\ge 2$. 
	Let $\pfeil^{\le k} = \bigcup_{0 \le j\le k} \pfeil^j$ for every $k\ge 0$.
	The \emph{transitive closure} is given by $\pfeil^*= \bigcup_{k \ge 0} \pfeil^k$. 
\end{defn}
The following definition shows how any GTS can be interpreted as a~TS. 
\begin{defn}[graph transition system] 
	Let $\R$ be a GTS and $\GG$ a set of graphs which is closed under rule application of $\R$. 
	The \emph{graph transition system} w.r.t.\ $\R$ and $\GG$ is the transition system $\tuple{\GG, \dder_\R}$. 
     A graph transition system $\tuple{\GG, \dder_\R}$ \emph{is of bounded path length} if $\sup_{\G \in \GG} \ell(G) < \infty$.  
\end{defn} 
\begin{exmp}[GTS of bounded path length] The rules $\tuple{\varnothing \pmor \begin{tikzpicture}[baseline=(center)]
		\renewcommand{\xdis}{7mm}
		\node[circle, draw=black, inner sep=0pt, minimum size=4mm] (product) {{\footnotesize $ A $}};
\node (center) at (product) [yshift=-1mm] {};
		\end{tikzpicture}}$ and $\tuple{\begin{tikzpicture}[baseline=(center)]
		\renewcommand{\xdis}{7mm}
		\node[circle, draw=black, inner sep=0pt, minimum size=4mm] (product) [label={[inner sep=0.4pt]-180:{\scriptsize $ 1 $}}]{{\footnotesize $ A $}};	
\node (center) at (product) [yshift=-1mm] {};
		\end{tikzpicture} \pmor \begin{tikzpicture}[baseline=(center)]
		\renewcommand{\xdis}{7mm}
		\node[circle, draw=black, inner sep=0pt, minimum size=4mm] (product) [label={[inner sep=0.4pt]-180:{\scriptsize $ 1 $}}]{{\footnotesize $ A $}};
		\node[circle, inner sep=0pt, draw=black, fill=black, minimum size=1.5mm] at (product) [xshift=5mm] (tokenP) {};
\node (center) at (product) [yshift=-1mm] {};
		\draw[-to]
		(product) edge (tokenP)
		;
		\end{tikzpicture}}$ together with the set of disjoint unions of unboundedly many (possibly non-isomorphic) star-shaped graphs  forms a graph transition system of bounded path length. 
\end{exmp}
\begin{rem}Note that we only demand bounded path length. If we additionally demand a bound on the node degree, 
the number of nodes/edges in each connected component of any graph in the graph class is bounded. 
This can be shown by an induction over the bound on the path length. 
\end{rem}
Often we are interested in the predecessors or successors of a given set of states in a transition system.
\begin{defn}[pre- \& postsets] 
	Let $\tuple{\St,\pfeil}$ be a transition system.
	For $\St' \subseteq \St$ and $k\ge 0$, 
	we define $\pre^k (\St') =\{s \in \St \fdg \exists s' \in \St': s \pfeil^k s' \}$ 
	and $\post^k (\St') = \{s \in \St \fdg \exists s' \in \St' :s' \pfeil^k s \} $. 
	Let $\pre^* (\St') = \bigcup_{k \ge 0} \pre^k(\St')$ and
	$\post^* (\St') = \bigcup_{k \ge 0} \post^k(\St')$.
	We abbreviate $ \post^1(\St') $ by $ \post(\St') $ and $\pre^1 (\St')$ by $\pre(\St)$.
 \end{defn}
GTSs, when interpreted as TSs, in general have an infinite state space.

\subsection{Well-structuredness} \label{subs:wsts}
While several problems are undecidable for transition systems in general
due to their infinite state space,
many interesting decidability results can be achieved if 
the system is \emph{well-structured} \cite{FinkelS01,AbdullaCJT96,Koenig17}.
\begin{defn}[well-quasi-order] 
	A \emph{well-quasi-order (wqo)} over a set $ X $ 
	is a quasi-order (a reflexive, transitive relation) 
	$\wqo \subseteq X \times X $ s.t.\
	every infinite sequence $\tuple{x_0, x_1, \ldots}$ in $ X $ 
	contains an increasing pair $x_i \wqo x_j$ with $i <j$. 
\end{defn}
We give two examples for wqos on graphs. In our setting, the subgraph order is of crucial importance.
\begin{exmp}[subgraph \& minor order]\label{ex:WQOsOnGraphs} ~
\begin{enumerate}[(i)]
\item	The subgraph order $ \wqo $ is given by $G\wqo H$ iff there is a total injective morphism $G \injto H$.
	Let $\GG_\ell$ be a graph class of bounded path length (with bound $ \ell $). 
	The restriction of $\wqo$ to $\GG_\ell$ is a
	wqo \cite{Koenig17,Ding92}. However, it is not a wqo on all graphs: consider, e.g.,\ the infinite sequence $\tuple{\trianglegraphUN,\squaregraph,\pentagraph, \ldots }$ of cyclic graphs of increasing length, which contains no increasing pair.
\item The minor order $ \mqo $ is given by $G \mqo H$ iff $G$ can be obtained from $H$
by a sequence of edge contractions, node and edge deletions. 
The minor order is a wqo on all graphs \cite{Koenig17,Robertson04}. \end{enumerate}
\end{exmp}
\begin{ass}
	From now on, we implicitly equip every set of graphs with the subgraph order.
	By $\wqo$ we mean either an abstract wqo or the subgraph order, depending on the context.
\end{ass}
\begin{defn}[closure \& basis]
	Let $ X $ be a set and $ \leq $ a wqo on $ X $.
	For every subset $ A $ of $ X $, we denote by 
	$\clo{A}=\{ x \in X \fdg \exists a \in A : a \wqo x \}$ the \emph{upward-closure} and 
	$\dlo{A}=\{ x \in X \fdg \exists a \in A : x \wqo a \}$ the \emph{downward-closure} of $ A $. 
	If $ A=\clo{A} $, then 
	a \emph{basis} of $ A $ is a subset $ \bI\subseteq A $ s.t.\ 
	(i) $\bI$ \emph{generates} $A$, i.e., $\clo{\bI}=A$, and 
	(ii)~any two distinct elements in $ \bI $ are \emph{incomparable}, 
	i.e., $\forall b_1 , b_2 \in \bI : b_1 \neq b_2 \Rightarrow b_1 \not\wqo b_2$.
\end{defn}
Sets $ A $ satisfying $ A=\clo{A} $ are later called ideals.
For well-structuredness, we demand that the wqo yields a simulation of smaller states by larger states. 
This condition is called \emph{compatibility}.
\begin{defn}[well-structured transition systems] \label{wellstr}
	Let $\tuple{\St,\pfeil}$ be transition system and $\wqo$ a decidable wqo on $ \St $,
	i.e., for each two given states $s,s' \in \St$, it is decidable whether $s \wqo s'$.
	The tuple $\tuple{\St,\wqo,\pfeil}$ is a \emph{(strongly) well-structured transition system}, if
	\begin{enumerate}[(i)]
		\item  The wqo is \emph{(strongly) compatible} with the transition relation, 
		i.e., for all $s_1, s'_1, s_2 \in \St$ with $s_1 \wqo s'_1$ and $s_1 \pfeil s_2$, 
		there exists $s'_2 \in S$ with $ s_2\wqo s'_2 $ and $s_1' \pfeil^* s'_2$ (strongly: $s'_1 \pfeil^1 s'_2$).
		\item For every $s \in \St$, a basis of $\clo{\pre(\clo{\{s\}})}$ is computable.  
	\end{enumerate}
\end{defn}

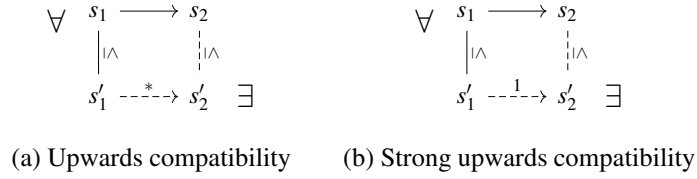
\begin{figure}[htb]\label{fig:UpwardsCompVisualization}
	\centering
	\begin{subfigure}[t]{0.3\textwidth}\centering
		\begin{tikzpicture}[scale=0.85,transform shape]
			\node[] (cd) {
				\begin{tikzcd}
				s_1 \arrow[r] \arrow[d, "\leq" labl, no head] & s_2 \arrow[d, "\leq" labl, no head, dashed] \\
				s_1' \arrow[r, "*", dashed]               & s_2'                                  
				\end{tikzcd}
			};
			\node at (cd.north west) [xshift=-2mm,yshift=-5mm] (A) {{\Large $ \forall $}};
			\node at (cd.south east) [xshift=2mm,yshift=5mm] (A) {{\Large $ \exists $}};
		\end{tikzpicture}
		\subcaption{Upwards compatibility}%
		\label{fig:UpwardsComp}
	\end{subfigure}
	\begin{subfigure}[t]{0.3\textwidth}\centering
		\begin{tikzpicture}[scale=0.85,transform shape]
			\node[] (cd) {
				\begin{tikzcd}
				s_1 \arrow[r] \arrow[d, "\leq"' labl, no head] & s_2 \arrow[d, "\leq" labl, no head, dashed] \\
				s_1' \arrow[r, "1", dashed]               & s_2'                                  
				\end{tikzcd}
			};
			\node at (cd.north west) [xshift=-2mm,yshift=-5mm] (A) {{\Large $ \forall $}};
			\node at (cd.south east) [xshift=2mm,yshift=5mm] (A) {{\Large $ \exists $}};
		\end{tikzpicture}
		\subcaption{Strong upwards compatibility}%
		\label{fig:StrongUpwardsComp}
	\end{subfigure}
	\caption{Visualization of the (strong) upwards compatibility property for transition systems.
	}
	\label{fig:WSTS-Vis}
\end{figure}

In Fig.~\ref{fig:WSTS-Vis}, both versions of compatibility are visualized.
The term \emph{(strongly) well-structured transition system} is often abbreviated by (S)WSTS.
In Sec.~\ref{sec:decidabilityResults}, we prove the decidability of resilience for SWSTSs. 
We include the definition of general WSTSs for clarity and to point out the differences.
Note that for GTSs, strong compatbility is achieved by applying the same (SPO) rule to bigger graphs. However,
in DPO, the bigger graph may not fullfill the dangling condition. Consider, e.g.,\ the rule which deletes a node. This rule
can be applied to the graph consisting of a single node but not to the graph \twonodesUN~in DPO. \par
The following result of K{\"o}nig \& St{\"u}ckrath terms sufficient conditions for GTSs to be well-structured.
\begin{lem}[\cite{Koenig17}] \label{Koenig}
Every graph transition system of bounded path length 
is strongly well-structured (equipped with the subgraph order).
\end{lem}
Note that in \cite{Koenig17}, K{\"o}nig \& St{\"u}ckrath consider labeled hypergraphs. 
However, the proof in this case is the same. 
The premise of bounded path length seems very restrictive, 
but we can still capture infinitely many graphs.
A usual example are graphs where the ``topology'' remains unchanged. 
It is also shown in \cite{Koenig17} that every \emph{lossy GTS} is well-structured 
w.r.t.\ the minor order and without restriction of the graph class. 
``Lossy'' means that every edge contraction rule is contained in the GTS. 
However, in this case, we do not obtain strong compatibility.

\begin{ass} 
	In the following, let $ \tuple{\St,\wqo, \pfeil} $
	be a strongly well-structured transition system.
\end{ass}
Upward- and downward-closed sets w.r.t.\ a given wqo are of special interest. 
Such sets are called ideals and used 
in Sec.~\ref{sec:adverseCondAndResProblems} to define resilience problems for WSTSs.
\begin{defn}[ideal] 
	An \emph{ideal} $\Ii \subseteq \St$ is an upward-closed set, 
	i.e., $\clo{\Ii}=\Ii$. 
	A \emph{bi-ideal} $\Ji\subseteq \St$ is an ideal which is also downward-closed, 
	i.e., $\clo{\Ji}=\Ji=\dlo{\Ji}$. 
	An \emph{anti-ideal} $\Ji \subseteq \St$ is a downward-closed set, 
	i.e.,\ $\dlo{\Ji}=\Ji$. The anti-ideal $\Ji$ is decidable if, given $s \in \St$, it is decidable whether $s \in \Ji$. 
\end{defn}

\begin{exmp}[ideal]
Let $\GG_\ell$ be a graph class of bounded path length. For every positive constraint $c$, $\Ii_c = \{ G \in \GG_\ell \,\vert\, G \models c \}$ is an ideal.
\end{exmp}
Bi-ideals often represent ``control states'' as in \cite{AbdullaCJT96}.
The notion of anti-ideal is the pendent to ideal. 
Since a downward-closed set does not have an ``upward-basis'' in general, 
we will demand that membership is decidable. 

\begin{exmp}[anti-ideal] 
Let $\GG_\ell$ be a graph class of bounded path length. For every negative constraint $c$, $\Ji_c = \{ G \in \GG \,\vert\, G \models c \}$ is a decidable anti-ideal.
\end{exmp}

The set of ideals of $\St$ is closed under preset, union, and intersection. 
\begin{fact}[stability of ideals] \label{stabideal}
	Let $\Ii,\Ji \subseteq \St$ be ideals. 
	Then the sets 
	$\pre(\Ii)$,
	$\Ii \cup \Ji$, and 
	$\Ii \cap \Ji$ are ideals.  
\end{fact}
A major point in our argumentation is the observation that 
every infinite ascending sequence of ideals w.r.t.\ a wqo eventually becomes stationary. 
 
\begin{lem}[\cite{AbdullaCJT96}] \label{Noether}
	For every infinite ascending sequence 
	$\tuple{\Ii_0 \subseteq \Ii_1 \subseteq \ldots}$ of ideals, 
	there exists a $k\ge 0$ s.t.\ $\Ii_{k}=\Ii_{k+1}$. 
	This directly implies $ \exists k_0\ge 0 \; \forall k\geq k_0 : \Ii_{k}=\Ii_{k_0} $.
\end{lem}  
Since ideals are in general infinite, 
we often want a finite representation. 
Similar to algebraic structures, 
ideals are represented by a finite basis (a minimal generating set).
Indeed, every ideal has a basis and every basis is finite.
We consider bases for complexity reasons. In theory, finite generating sets are sufficient to carry out our approach.
\begin{fact}[\cite{AbdullaCJT96}] \label{fbasis}
	(i) For every ideal $\Ii \subseteq \St$, 
	there exists a finite basis $\bI$ of $\Ii$. 
	(ii) Given a finite set $A \subseteq \St$ with $\Ii = \clo{A}$, 
	we can compute a finite basis $\bI$ of $\Ii$.
\end{fact}
\subsection{Ideal Reachability}\label{IdReach}
In \cite{AbdullaCJT96}, Abdulla et al.\ exploit Lemma~\ref{Noether} 
to show the decidability of \emph{ideal reachability} (also called \emph{coverability}) for strongly well-structured transition systems. 
The corresponding algorithm forms the basis of our results. 
We present its basic idea. 
For any ideal $\Ii$, another ideal $\Ii^*$ is constructed, 
s.t.\ $\exists s' \in \Ii : s \pfeil^* s'$ iff $s \in \Ii^*$. 
This is clearly the case for $\Ii^*=\pre^*(\Ii)= \bigcup_{j\ge 0} \pre^j (\Ii)$. 
The idea is to iteratively construct the sequence of the ideals $\Ii^k=\bigcup_{0\le j \le k} \pre^j (\Ii)$ until it becomes stable.
\begin{defn}[index]
For an ideal $\Ii \subseteq \St$ and $k\ge 0$, let $\Ii^k = \bigcup_{0\le j \le k} \pre^j (\Ii) \subseteq \Ii^{k+1}$. The \emph{index} $k(\Ii)$ is the smallest $k_0$ s.t.\ $\Ii^k = \Ii^{k_0}$ for all $k \ge k_0$.
\end{defn}
Lemma~\ref{Noether} ensures that $k(I)$ always exists. However, we have to show that $\Ii^k = \Ii^{k+1}$ implies $k(\Ii) \le k$ to obtain a stop condition.
This follows by the observation that $\Ii^{k+1}=\Ii \cup \pre(\Ii^k)$.
\begin{fact}[stop condition] \label{stopcond}
Let $\Ii \subseteq \St$ be an ideal and $k\ge 0$ s.t.\ $\Ii^k = \Ii^{k+1}$, then $\Ii^\ell = \Ii^k$ for all $\ell \ge k$, i.e.,\ $k(\Ii) \le k$. This also implies that $\pre^*(\Ii) = \Ii^k$.
\end{fact}

Since ideals are infinite, we cannot carry this construction out directly, but we use a basis for representing an ideal. If we can show the computability of a basis in every iteration step, we obtain an algorithm which can decide whether we can reach an ideal $\Ii$ from a given state $s$.

\begin{lem}[\cite{AbdullaCJT96}] \label{Abdulla} Given a basis of an ideal $\Ii \subseteq \St$, and a state $s$ of a strongly well-structured transition system, we can decide whether we can reach $\Ii$ from $s$. 
\end{lem}

\begin{proof}
We have to show that we can compute a basis of $\Ii^{k+1}$ if we are given a basis of $\Ii^k$. Then the decidability of the stop condition follows directly. Let $B$ be a basis of $\Ii^k$. 
We have 
\begin{equation*}
\Ii^{k+1}=\Ii \cup \pre(\Ii^k) = \Ii \cup \bigcup_{s' \in B} \pre(\clo{\{s'\}}).
\end{equation*} 
Since $ \pre(\clo{\{s'\}})$ is computable for any $s' \in \St$ by definition, we obtain a finite generating set of $\Ii^{k+1}$. By Fact \ref{fbasis}, we can compute a basis of $\Ii^{k+1}$. 
\end{proof}

%% file: sources/contribution.tex
\section{Adverse Conditions and Resilience Problems}\label{sec:adverseCondAndResProblems}
We put adverse conditions and resilience into context by using joint graph transformation
systems~\cite{Oezkan20}. 
Abstracting from the setting of GTSs,
we identify resilience problems for TSs.
\subsection{Joint Graph Transformation Systems}
We recapitulate the modeling of adverse conditions by joint graph transformation
systems, introduced in~\cite{Oezkan20}.
We define joint graph transformation systems, which involve a system and an environment, as well as an automaton 
modeling the interaction between them. 
Both, system and environment, are~GTSs.
\begin{ass} 
	In the following, let $\Lambda$ be a fixed label alphabet, 
	and $\Sys$ and $\Env$ be GTSs over $\Lambda$, called \emph{system} and \emph{environment}, respectively. W.l.o.g.,
	 we assume that $\Sys$ and $\Env$ are disjoint. 
	(If $\Sys$ and $\Env$ share a common rule $r$, 
	we assign $r$ different names in $\Sys$ and $\Env$.) 
\end{ass}
We specify the class of automata which are used to regulate 
the interaction between system and environment. 
These control automata are similar to $\omega$-automata, 
see, e.g.,\ \cite{Thomas90}.
\begin{defn}[control automaton] 
	A \emph{control automaton} of $\tuple{\Sys, \Env}$ is 
	a tuple $A=\tuple{Q, q_0, \delta, \select}$ consisting of a 
	finite set $Q$ disjoint from $\Lambda$, called the \emph{state set}, 
	an \emph{initial state} $q_0 \in Q$, a \emph{transition relation} 
	$\delta \subseteq Q \times Q$, and a function $\select: \delta \to \Pot ( \Sys \cup \Env ) $ 
	(into the power set of $\Sys \cup \Env$), called the \emph{selection function}. 
\end{defn}

A joint graph transformation system is obtained 
by \emph{synchronizing} the system, repectively, the environment, 
with the control automaton, 
and then joining both sets of enriched rules. 

\begin{defn}[joint graph transformation system] 
	Let $A=\tuple{Q,q_0 , \delta, \select }$ be a control automaton of $\tuple{\Sys,\Env}$. 
	The \emph{joint graph transformation system} of $\Sys$ and $\Env$ w.r.t.\ $A$ 
	is the graph transformation system $\Sys_A \cup  \Env_A$
	where for a rule set $\R \in \{\Sys, \Env\}$, the \emph{enriched rule set} $\R_A$ is given by
	\begin{equation*} 
		\R_A = \{ \tuple{L,q} \pmor \tuple{R,q'}\,\vert\, 
		\tuple{q,q'} \in \delta\text{ and } 
		\tuple{L \pmor R} \in \R \cap \select \tuple{q,q'} 
		\},
	\end{equation*} 
	and for a graph $G$ and a state $q$, the tuple $\tuple{G,q}$ denotes the disjoint union of $G$ and a node labeled with $q$.
In the partial morphism $\tuple{L,q} \pmor \tuple{R,q'}$, the node labeled with $q$ is mapped to the node labeled with~$q'$. 
	
\end{defn}

We refine our notion of joint graph transformation systems, 
namely to \emph{annotated joint graph transformation systems}, 
which also carry the information whether the last applied rule was a system or environment rule. 
This is realized by a node labeled with ``$\ssys$'' or ``$\senv$''.  
\begin{nota} 
	For a joint graph transformation system $\Sys_A \cup  \Env_A$, 
	the symbol $\m(\Sys)=\ssys$ or $\m(\Env)=\senv$, 
	is the \emph{marker} of $\Sys$ or $\Env$, respectively. 
	For a rule $r \in \R$ and $\R \in \{\Sys, \Env\}$, 
	let $\m(r)=\m(\R)$ be the marker of~$r$. 
	The set of all markers $\Mrk= \{\top, \ssys, \senv\}$ includes also the symbol
	$\top$, usually indicating a start graph. 
\end{nota}

For the explicit construction, we can use \emph{premarkers} to reduce the number of rules. 
For a more extensive account on this technical detail, consult \cite{Oezkan20}.

\begin{defn}[annotated joint graph transformation system] 
	Let $\Sys_A \cup  \Env_A$ be a joint graph transformation systems w.r.t.\ 
	a control automaton $A$ of $\tuple{\Sys,\Env}$. 
	The \textit{annotated joint graph transformation system} of $\Sys$ and $\Env$ w.r.t.\ $A$ 
	is $\Sys_A' \cup  \Env_A' $, 
	where for a rule set $\R \in \{\Sys,\Env\}$, the \textit{marked rule set} $\R'_A$ is defined as
	\begin{equation*}
	\R'_A = \{ \tuple{L,q,m} \pmor \tuple{R,q',m'} \,\vert\, 
	\tuple{L,q} \pmor \tuple{R,q'} \in \R_A,\, m \in \Mrk,\, m' =\m( \R )\} ,
	\end{equation*}
	where $\tuple{G,q,m}$ in turn denotes the disjoint union of a graph $G$, a node labeled with a state $q$, and a node labeled with a marker $m$.
	In the partial morphism $\tuple{L,q,m} \pmor \tuple{R,q',m'}$, the node lableled with~$m$ is mapped to the node labeled with $m'$. 
\end{defn}
We explicate the state set of annotated joint GTSs. These graphs are of the form $\tuple{G,q,m}$ for a state~$q$ of the control automaton and a marker $m$. 
We denote a class of all such graphs by $\GG \oplus Q \oplus \Mrk$. Using such graphs instead of the product of graphs we can directly apply the result of \cite{Koenig17} for GTSs (Lemma~\ref{Koenig}).
\begin{defn}[joint graph transition system]
Let $(\Sys_A \cup  \Env_A)'$ be an annotated joint GTS and $\GG'$ be a class of graphs which is of the form $\GG \oplus Q \oplus \Mrk$ and closed under rule application of $(\Sys_A \cup  \Env_A)'$. The 
graph transition system $\tuple{\GG', \trafo_{(\Sys_A \cup  \Env_A)'}}$ is called \emph{annotated joint graph transition system}.
\end{defn}
Note that we usually begin our analysis at a start graph of the form $\tuple{G,q_0 , \top}$.

\begin{exmp}[supply chain]\label{prodch}
We model a simple supply chain
with graph transformation rules. The infrastructure (topology) is given in the following start graph: 
\begin{center}\begin{tikzpicture}[baseline=(product.center)]
	\renewcommand{\xdis}{8mm}
	\renewcommand{\ydis}{3.5mm}
	\node[circle, draw=black, inner sep=0pt, minimum size=4mm] (product) {{\footnotesize $ P $}};
	\node[circle, draw=black, inner sep=0pt, minimum size=4mm, right=of product] (warehouse) {{\footnotesize $ W $}};
	\node[circle, draw=black, inner sep=0pt, minimum size=4mm, above right=of warehouse] (s1) {{\footnotesize $ S_1 $}};
	\node[circle, draw=black, inner sep=0pt, minimum size=4mm, below right=of warehouse] (s2)  {{\footnotesize $ S_2 $}};
	
	\node[circle, inner sep=0pt, draw=black, fill=black, minimum size=1.5mm] at (warehouse) [yshift=4mm] (tokenP) {};
	\node[circle, inner sep=0pt, draw=black, fill=black, minimum size=1.5mm] at (s1) [xshift=4mm] (tokenP1) {};
	\node[circle, inner sep=0pt, draw=black, fill=black, minimum size=1.5mm] at (s2) [xshift=4mm] (tokenP2) {};
	\draw[-]
		(product) edge (warehouse)
		(warehouse) edge (tokenP)
		(warehouse) edge (s1)
		(warehouse) edge (s2)
		(s1) edge (tokenP1)
		(s2) edge (tokenP2)
	;
\end{tikzpicture}
\end{center}
A production site ($ P $) is connected to a warehouse ($ W $)
which again is connected to two stores $ S_1 $ and~$ S_2 $. 
Each black node indicates one product at the corresponding (connected) location. 
The behavior in this production chain is modeled by the graph transformation rules in Fig.~\ref{fig:ProductChainGTS}.
The system rules consists of 
$\mathtt{pr}$ (the completion of a product at the production site $ P $),
$\mathtt{tr}$ (transporting a product from~$ P $ to the warehouse $ W $),
and $\mathtt{sh_1}$ and $\mathtt{sh_2}$ (shipping a product from $ W $ to one of the two stores $ S_1 $, $ S_2 $).
The environment rules describe external impacts.
Namely,
$\mathtt{ac}$ describes an accident in the warehouse which leads to the loss of one product,
and $\mathtt{b_1}$ and $\mathtt{b_2}$ describe that a product is bought from $ S_1 $ or $ S_2 $, respectively.
\begin{figure}[htb]
	\begin{subfigure}[b]{0.71\textwidth}
		\centering
		$\Sys \left\{ 
		\begin{array}{ll} 
		\mathtt{pr:} &
		\left\langle
		\begin{tikzpicture}[baseline=(base.center)]
		\renewcommand{\xdis}{10mm}
		\node[circle, draw=black, inner sep=0pt, minimum size=4mm] (product) [label={[inner sep=0.5pt]left:{\scriptsize $ 1 $}}] {{\footnotesize $ P $}};
		\node[] at ([yshift=-1mm]product.center) (base) {};
		\end{tikzpicture}
		~\rightharpoonup~ 
		\begin{tikzpicture}[baseline=(base.center)]
		\renewcommand{\xdis}{10mm}
		\node[circle, draw=black, inner sep=0pt, minimum size=4mm] (product) [label={[inner sep=0.5pt]left:{\scriptsize $ 1 $}}] {{\footnotesize $ P $}};
		\node[circle, inner sep=0pt, draw=black, fill=black, minimum size=1.5mm] at (product) [xshift=4mm] (tokenP) {};
		\node[] at ([yshift=-1mm]product.center) (base) {};
		\draw[-]
		(product) edge (tokenP)
		;
		\end{tikzpicture} 
		\right\rangle\vspace*{1mm}\\ 
		
		\mathtt{tr:} &
		\left\langle
		\begin{tikzpicture}[baseline=(base.center)]
		\renewcommand{\xdis}{7mm}
		\node[circle, draw=black, inner sep=0pt, minimum size=4mm] (product)  [label={[inner sep=0.5pt]-10:{\scriptsize $ 1 $}}] {{\footnotesize $ P $}};
		\node[circle, draw=black, inner sep=0pt, minimum size=4mm, right=of product] (warehouse)  [label={[inner sep=0.5pt]-10:{\scriptsize $ 2 $}}]{{\scriptsize $ W $}};
		\node[circle, inner sep=0pt, draw=black, fill=black, minimum size=1.5mm] at (product) [xshift=-4mm] (tokenP)  [label={[inner sep=0.5pt]left:{\scriptsize $ 3 $}}] {};
		\node[] at ([yshift=-1mm]product.center) (base) {};
		\draw[-]
		(product) edge (warehouse)
		(product) edge (tokenP)
		;
		\end{tikzpicture}
		~\rightharpoonup~ 
		\begin{tikzpicture}[baseline=(base.center)]
		\renewcommand{\xdis}{7mm}
		\node[circle, draw=black, inner sep=0pt, minimum size=4mm] (product)  [label={[inner sep=0.5pt]-10:{\scriptsize $ 1 $}}] {{\footnotesize $ P $}};
		\node[circle, draw=black, inner sep=0pt, minimum size=4mm, right=of product] (warehouse)  [label={[inner sep=0.4pt]-25:{\scriptsize $ 2 $}}] {{\footnotesize $ W $}};
		\node[circle, inner sep=0pt, draw=black, fill=black, minimum size=1.5mm] at (warehouse) [xshift=4mm] (tokenW)  [label={[inner sep=0.5pt]right:{\scriptsize $ 3 $}}] {};
		\node[] at ([yshift=-1mm]product.center) (base) {};
		\draw[-]
		(product) edge (warehouse)
		(warehouse) edge (tokenW)
		;
		\end{tikzpicture} 
		\right\rangle  \vspace*{1mm}\\ 
		
		\mathtt{sh_1:} &
		\left\langle
		\begin{tikzpicture}[baseline=(base.center)]
		\renewcommand{\xdis}{7mm}
		\node[circle, draw=black, inner sep=0pt, minimum size=4mm] (product)  [label={[inner sep=0.5pt]-10:{\scriptsize $ 1 $}}] {{\footnotesize $ W $}};
		\node[circle, draw=black, inner sep=0pt, minimum size=4mm, right=of product] (warehouse)  [label={[inner sep=0.5pt]-10:{\scriptsize $ 2 $}}] {{\footnotesize $ S_1 $}};
		\node[circle, inner sep=0pt, draw=black, fill=black, minimum size=1.5mm] at (product) [xshift=-4mm] (tokenP)  [label={[inner sep=0.5pt]left:{\scriptsize $ 3 $}}] {};
		\node[] at ([yshift=-1mm]product.center) (base) {};
		\draw[-]
		(product) edge (warehouse)
		(product) edge (tokenP)
		;
		\end{tikzpicture}
		~\rightharpoonup~ 
		\begin{tikzpicture}[baseline=(base.center)]
		\renewcommand{\xdis}{7mm}
		\node[circle, draw=black, inner sep=0pt, minimum size=4mm] (product) [label={[inner sep=0.5pt]-10:{\scriptsize $ 1 $}}] {{\footnotesize $ W $}};
		\node[circle, draw=black, inner sep=0pt, minimum size=4mm, right=of product] (warehouse) [label={[inner sep=0.4pt]-25:{\scriptsize $ 2 $}}] {{\footnotesize $ S_1 $}};
		\node[circle, inner sep=0pt, draw=black, fill=black, minimum size=1.5mm] at (warehouse) [xshift=4mm] (tokenW) [label={[inner sep=0.5pt]right:{\scriptsize $ 3 $}}] {};
		\node[] at ([yshift=-1mm]product.center) (base) {};
		\draw[-]
		(product) edge (warehouse)
		(warehouse) edge (tokenW)
		;
		\end{tikzpicture} 
		\right\rangle \vspace*{1mm}\\ 
		
		\mathtt{sh_2:} &
		\left\langle
		\begin{tikzpicture}[baseline=(base.center)]
		\renewcommand{\xdis}{7mm}
		\node[circle, draw=black, inner sep=0pt, minimum size=4mm] (product) [label={[inner sep=0.5pt]-10:{\scriptsize $ 1 $}}] {{\footnotesize $ W $}};
		\node[circle, draw=black, inner sep=0pt, minimum size=4mm, right=of product] (warehouse) [label={[inner sep=0.5pt]-10:{\scriptsize $ 2 $}}] {{\footnotesize $ S_2 $}};
		\node[circle, inner sep=0pt, draw=black, fill=black, minimum size=1.5mm] at (product) [xshift=-4mm] (tokenP) [label={[inner sep=0.5pt]left:{\scriptsize $ 3 $}}] {};
		\node[] at ([yshift=-1mm]product.center) (base) {};
		\draw[-]
		(product) edge (warehouse)
		(product) edge (tokenP)
		;
		\end{tikzpicture}
		~\rightharpoonup~ 
		\begin{tikzpicture}[baseline=(base.center)]
		\renewcommand{\xdis}{7mm}
		\node[circle, draw=black, inner sep=0pt, minimum size=4mm] (product)  [label={[inner sep=0.5pt]-10:{\scriptsize $ 1 $}}] {{\footnotesize $ W $}};
		\node[circle, draw=black, inner sep=0pt, minimum size=4mm, right=of product] (warehouse) [label={[inner sep=0.4pt]-25:{\scriptsize $ 2 $}}] {{\footnotesize $ S_2 $}};
		\node[circle, inner sep=0pt, draw=black, fill=black, minimum size=1.5mm] at (warehouse) [xshift=4mm] (tokenW)  [label={[inner sep=0.5pt]right:{\scriptsize $ 3 $}}] {};
		\node[] at ([yshift=-1mm]product.center) (base) {};
		\draw[-]
		(product) edge (warehouse)
		(warehouse) edge (tokenW)
		;
		\end{tikzpicture} 
		\right\rangle 
	\end{array} 
	\right.$\hspace*{-0.5mm}\mbox{,  
	}$\Env \left\{ 
	\begin{array}{ll}
		\mathtt{ac:} &
		\left\langle
		\begin{tikzpicture}[baseline=(base.center)]
		\renewcommand{\xdis}{10mm}
		\node[circle, draw=black, inner sep=0pt, minimum size=4mm] (product) [label={[inner sep=0.4pt]left:{\scriptsize $ 1 $}}] {{\footnotesize $ W $}};
		\node[circle, inner sep=0pt, draw=black, fill=black, minimum size=1.5mm] at (product) [xshift=4mm] (tokenP)  {};
		\node[] at ([yshift=-1mm]product.center) (base) {};
		\draw[-]
		(product) edge (tokenP)
		;
		\end{tikzpicture} 
		~\rightharpoonup~ 
		\begin{tikzpicture}[baseline=(base.center)]
		\renewcommand{\xdis}{10mm}
		\node[circle, draw=black, inner sep=0pt, minimum size=4mm] (product) [label={[inner sep=0.4pt]left:{\scriptsize $ 1 $}}] {{\footnotesize $ W $}};
		\node[] at ([yshift=-1mm]product.center) (base) {};
		\end{tikzpicture}
		\right\rangle\vspace*{1mm}\\ 
		
		\mathtt{b_1:} &
		\left\langle
		\begin{tikzpicture}[baseline=(base.center)]
		\renewcommand{\xdis}{10mm}
		\node[circle, draw=black, inner sep=0pt, minimum size=4mm] (product) [label={[inner sep=0.4pt]left:{\scriptsize $ 1 $}}] {{\footnotesize $ S_1 $}};
		\node[circle, inner sep=0pt, draw=black, fill=black, minimum size=1.5mm] at (product) [xshift=4mm] (tokenP)  {};
		\node[] at ([yshift=-1mm]product.center) (base) {};
		\draw[-]
		(product) edge (tokenP)
		;
		\end{tikzpicture} 
		~\rightharpoonup~ 
		\begin{tikzpicture}[baseline=(base.center)]
		\renewcommand{\xdis}{10mm}
		\node[circle, draw=black, inner sep=0pt, minimum size=4mm] (product) [label={[inner sep=0.4pt]left:{\scriptsize $ 1 $}}] {{\footnotesize $ S_1 $}};
		\node[] at ([yshift=-1mm]product.center) (base) {};
		\end{tikzpicture}
		\right\rangle\vspace*{1mm}\\ 
		
		\mathtt{b_2:} &
		\left\langle
		\begin{tikzpicture}[baseline=(base.center)]
		\renewcommand{\xdis}{10mm}
		\node[circle, draw=black, inner sep=0pt, minimum size=4mm] (product) [label={[inner sep=0.4pt]left:{\scriptsize $ 1 $}}]{{\footnotesize $ S_2 $}};
		\node[circle, inner sep=0pt, draw=black, fill=black, minimum size=1.5mm] at (product) [xshift=4mm] (tokenP) {};
		\node[] at ([yshift=-1mm]product.center) (base) {};
		\draw[-]
		(product) edge (tokenP)
		;
		\end{tikzpicture} 
		~\rightharpoonup~ 
		\begin{tikzpicture}[baseline=(base.center)]
		\renewcommand{\xdis}{10mm}
		\node[circle, draw=black, inner sep=0pt, minimum size=4mm] (product) [label={[inner sep=0.4pt]left:{\scriptsize $ 1 $}}] {{\footnotesize $ S_2 $}};
		\node[] at ([yshift=-1mm]product.center) (base) {};
		\end{tikzpicture}
		\right\rangle
	\end{array}
	\right.$
	\caption{The two GTSs $ \Sys $ (system) and $ \Env $ (environment)}
	\label{fig:ProductChainGTS}
	\end{subfigure}
	\begin{subfigure}[b]{0.28\textwidth}
		\centering
		\begin{tikzpicture}
		\renewcommand{\xdis}{1.7cm}
		\renewcommand{\ydis}{0.6cm}
		\node[graphnode,fill=white,scale=2,inner sep=1pt] (init) {$ e $};
		\node[graphnode, right=of init] (d) {};
		\node[graphnode, right=of d] (dd) {};
		\node[graphnode,above=of d] (pt) {};
		\node[graphnode] at ($ (init)!0.5!(pt) $) (p) {};
		\node[graphnode,above=of dd] (ptpt) {};
		\node[graphnode] at ($ (pt)!0.5!(ptpt) $) (ptp) {};
		\node[graphnode] at ($ (d)!0.5!(dd) $)  (dp) [yshift=-2*\ydis] {};
		
		\node at (init) [xshift=-6mm] (root) {};
		\draw[->]
		(root)	edge	(init)
		(init)	edge node[above,sloped,inner sep=1pt]{\footnotesize$\mathtt{pr}$} (p)
		(p)	edge node[above,sloped,inner sep=1pt]{\footnotesize$\mathtt{tr}$} (pt)
		(pt)	edge node[above,sloped,inner sep=1pt]{\footnotesize$\mathtt{pr}$} (ptp)
		(ptp)	edge node[above,sloped,inner sep=1pt]{\footnotesize$\mathtt{tr}$} (ptpt)
		(d)	edge node[below,sloped,inner sep=1pt]{\footnotesize$\mathtt{pr}$} (dp)
		(dp)	edge node[below,sloped,inner sep=1pt]{\footnotesize$\mathtt{tr}$} (dd)
		(init)	edge node[below,sloped,inner sep=1pt]{\footnotesize$\mathtt{sh_i}$} (d)
		(d)	edge node[below,sloped,inner sep=1pt]{\footnotesize$\mathtt{sh_i}$} (dd)
		(pt)	edge node[below,sloped,inner sep=1pt,pos=0.3]{\footnotesize$\mathtt{sh_i}$} (dd)
		;
		\draw[->, rounded corners]
		(ptpt) -- ++(0,0.4cm) -- node[above,inner sep=1pt, pos=0.6]{\footnotesize$\mathtt{ac}$, $ \mathtt{b_i} $} ++(-2*\xdis,0) -- (init)
		;
		\draw[->, rounded corners]
		(dd) -- ++(0,-0.8cm) -- node[above,inner sep=1pt, pos=0.6]{\footnotesize $ \mathtt{b_i} $} ++(-2*\xdis,0) -- (init)
		;
		\end{tikzpicture}
		\caption{Control automaton}
		\label{fig:ProductChainAutomaton}
	\end{subfigure}
	\caption{A joint GTS consisting of a GTS for system and environment, each, and a control automaton.}
	\label{fig:ProductChainJointGTS}
\end{figure}
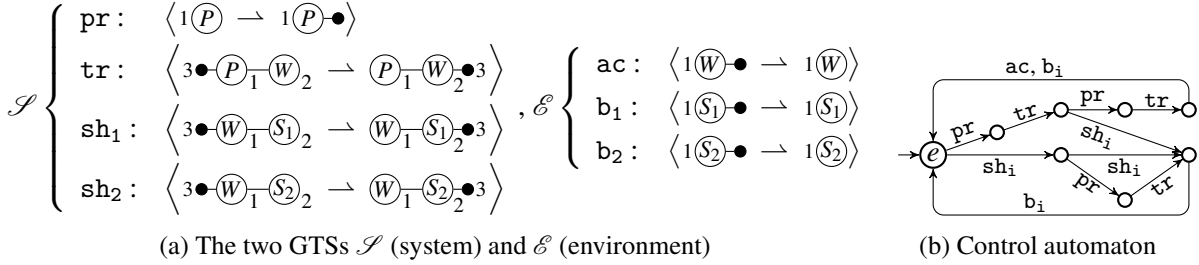
The control automaton in Fig.~\ref{fig:ProductChainAutomaton}
describes the possible order of rule applications. 
We are interested in the question when
the product is again \emph{in stock} (at least $1$ product in the warehouse and in each of both stores) whenever a customer buys a product or when an accident
in the warehouse happens. 
After each such transition, the automaton is in the state~$ e $.
Regardless of the current situation, in $17$ steps we can accomplish that the product is in stock by first producing 
and transporting $6$ products with a following accident ($ 3 $ products will get lost) and shipping them to the stores afterwards. However, what is the minimal number of steps
in which we can reach a situation where the product is in stock whenever someone bought a product or a product got lost in an accident?  
\end{exmp}
We come back to that question in Ex.~\ref{prodchPN} in Sec.~\ref{sec:PN}. We describe the setting for joint GTSs which we investigate: 
Consider a safety condition $c$, given as positive constraint, and the set of graphs 
${\Ii_c =}$ ${\{G' \in \GG \oplus Q \oplus \Mrk \,\vert\, G' \models c \}}$ which satisfy $c$.
Similarly, let $\Ji_\senv = \{ G' \in \GG \oplus Q \oplus \Mrk \,\vert\, G' \models \exists \senv \}$
(all graphs obtained by an environment interference; $ \exists \senv$ means that there 
exists a node labeled with $\senv$).
The environment is usually modeled in a such way that it has an adverse effect on 
the satisfaction of $c$. 
\emph{Resilience} in this context means
that the system can withstand such an
adverse condition.
We ask whether we can reach a graph in $\Ii_c$
in a reasonable amount of time whenever we reach a graph in $\Ji_\senv$. 
By a ``reasonable amount of time'', we mean either that a number $k$ of steps is given in 
which $\Ii_c$ should be reached (\emph{explicit resilience}), 
or that $\Ii_c$ should be reached in a bounded number of steps (\emph{bounded resilience}).

Another approach is to consider the set $\Ji_{\neg c}= \{G' \in \GG \oplus Q \oplus \Mrk \vert G' \not\models c \}$
instead of $\Ji_\senv$. So, we ask whether we can reach a graph which satisfies $c$ in a bounded amount of 
time/in at most $k$ steps whenever we reach a graph which does not satisfy $c$, i.e., an error state. Both instances of the problem are reasonable,
and if we can give a positive answer for the latter one, we can also give a positive answer for the first one.
We focus on the first problem (adverse conditions), but the results we obtain in Sec.~\ref{sec:decidabilityResults}
abstract from a specific $ \Ji $ and therefore
also apply to the latter one (error states).

\subsection{Abstract Resilience Problems}
The previous motivation gives rise to a more abstract definition of resilience problems, namely in the framework of TSs.
Recall that, when we explicate a state set, every GTS can be interpreted~as~a~TS. 

We assume that a TS $\tuple{\St, \pfeil}$ comes along with a set of propositions each of which is either satisfied or not 
satisfied by each state of the TS. Let $\safe$ (\emph{safety} condition) and $\envi$ (\emph{bad} condition) be propositions.
Note that $\envi$ is not necessarily equivalent to $\neg \safe$.   
We ask whether we can reach a state which satisfies $\safe$ in a reasonable amount of time whenever we reach a state which 
satisfies $\envi$. From this we formulate two resilience problems. 
First consider the case where the recovery time is bound by a natural number $k \ge 0$, i.e.,\
the \emph{(abstract) explicit resilience problem}.
\begin{problem}[lined]{Explicit Resilience Problem}
	Given: & A state $s$ of a TS $ \tuple{\St,\pfeil} $, propositions $\safe$ and $\envi$, a natural number $k \ge 0$.\\
	Question: & $ \forall s'\in\St: (s'\models \envi \land s \pfeil^* s') \Rightarrow \exists s''\in \St : s' \pfeil^{\le k} s'' \land  s'' \models \safe$ ?
\end{problem}
If we assume that the transition system yields infinite sequences of transitions, we can express the property to be evaluated in CTL by
$s \models \opAG (\envi \impl \bigvee_{0 \le j \le k} \opEX^j \safe )$.
We can also ask whether there exists such a bound $k$. We call this problem the \emph{(abstract) bounded resilience problem}. 
\begin{problem}[lined]{Bounded Resilience Problem}
	Given: & A state $s$ of a TS $ \tuple{\St,\pfeil} $, propositions $\safe$ and $\envi$.\\
	Question: & $\exists k \ge 0\  \forall s'\in\St: (s'\models \envi \land s \pfeil^* s') \Rightarrow \exists s''\in \St : s' \pfeil^{\le k} s'' \land  s'' \models \safe$ ?
\end{problem}
Both problems are undecidable: For $\safe=\false$, resilience is equivalent to reachability of $\envi$. 
\section{Decidability Results}\label{sec:decidabilityResults}
Many interesting decidability results can be obtained if we assume that 
a transition system is well-structured \cite{AbdullaCJT96,FinkelS01,Koenig17}. 
We formulate the resilience problems from the previous section for WSTSs
and show decidability of both, the explicit and the bounded resilience problem,
in the setting of SWSTSs.
\subsection{Resilience Problems in a Well-structured Framework}
Properties in well-structured transition systems are often given as upward- or downward closed sets~\mbox{\cite{AbdullaCJT96,FinkelS01}}. 
Ideals enjoy suitable features for verification such as finite representation and stability, and anti-ideals are their complements (cp. Sec.~\ref{subs:wsts}). 
Transfering the abstract resilience problems into this framework, it is therefore reasonable to
demand that both propositions, $ \safe $ and $ \envi $, are given by ideals or anti-ideals. 
For our purpose, the following setting suits very well: we assume that 
the safety property is given by an ideal and the bad condition by a decidable anti-ideal. 

From these considerations, we formulate ``instances'' of the abstract resilience
problems for well-struc\-tured transition systems. 
Again, we first consider the case where the recovery time is bounded by a $ k\in\mathbb{N} $,
the \emph{explicit resilience problem for WSTSs}.
\begin{problem}[lined]{Explicit Resilience Problem for WSTSs}
	Given: & A state $s$ of a WSTS $ \tuple{\St,\wqo,\pfeil} $, a basis of $\clo{\post^*(s)}$, an ideal $\Ii$ with a given basis, a decidable anti-ideal~$\Ji$, a natural number $k \ge 0$.\\
	Question: & $ \forall s'\in\Ji:  (s \pfeil^* s') \Rightarrow \exists s''\in \Ii : s' \pfeil^{\le k} s'' $ ?
\end{problem}
Analogously, we formulate the 
\emph{bounded resilience problem for WSTSs}.
\begin{problem}[lined]{Bounded Resilience Problem for WSTSs}
	Given: & A state $s$ of a WSTS $ \tuple{\St,\wqo,\pfeil} $, a basis of $\clo{\post^*(s)}$, an ideal $\Ii$ with a given basis, a decidable anti-ideal~$\Ji$.\\
	Question: & $\exists k \ge 0\  \forall s'\in\Ji: (s \pfeil^* s') \Rightarrow \exists s''\in \Ii : s' \pfeil^{\le k} s'' $ ?
\end{problem} 
From now on, we mean one of the previously defined resilience problems for WSTSs if we speak of a resilience problem. 
If the answer of the bounded (explicit) resilience problem is positive, we say that $\tuple{\St,\wqo,\pfeil}$ is \emph{resilient} \emph{(k-step resilient)} w.r.t.\ $\Ii$ and $\Ji$ starting from $s$. In this context, $s$ is a \emph{start state}.

\begin{rem} 
The premise that a basis of $\clo{\post^*(s)}$ is given is a strong but reasonable assumption. In general, 
we cannot simply compute  
the sequence of ideals $P_k= \bigcup_{0 \le j \le k } \clo{\post^j (s)}$ until it becomes stationary. This sequence does become
stationary by Lemma~\ref{Noether}. However, in contrast to the case in Lemma~\ref{Abdulla},
$P_{k+1}=P_k$ is not a sufficient stop condition. So, this way it is not algorithmically checkable when
we have reached~$k_0$ s.t.\ $P_\ell =P_{k_0}$ for every $\ell \ge k_0$.
However, 
we investigate resilience of GTSs each of which constitutes a SWSTS. 
A sufficient condition for strong well-structuredness is boundedness of the path length (cp.\ Lemma~\ref{Koenig}). This holds, e.g.,\ 
for graph classes where the ``topology'' is static. For these graph classes, a basis of all successors is often easier to determine than in general.
A typical example for such GTSs are Petri nets, where such a basis is computable (Sec.~\ref{sec:PN}).
In Sec.~\ref{sec:approximations}, we drop the assumption, and show that we can still approximate a basis of $\clo{\post^*(s)}$
to achieve approximation results for resilience.
\end{rem}
\subsection{Decidability}\label{sec:decidability}
Abdulla et al.\ show in \cite{AbdullaCJT96} that ideal reachability is decidable for SWSTSs (cp.\ Lemma~\ref{Abdulla}).
In \cite{FinkelS01},
\mbox{Finkel \& Schnoebelen} show that ideal reachability (or \emph{coverability}) is 
also decidable for WSTSs. 
Both algorithms coincide in the case of strong well-structuredness.
K{\"o}nig \& St{\"u}ckrath \cite{Koenig17} use the algorithm
of \cite{FinkelS01} for the \emph{backwards analysis} for (generalized) well-structured GTSs. 

The main difference between the algorithms in \cite{AbdullaCJT96} and \cite{FinkelS01} is that
for (not necessarily strongly) WSTSs, $\pre( \Ii')$  in general, for any ideal $\Ii'$,
is not an ideal. Thus, Finkel \& \mbox{Schnoebelen} consider in every iteration step the ideal $\clo{\pre( \Ii')}$ instead of $\pre( \Ii')$.
Now the same arguments like before hold (cp. Sec.~\ref{IdReach}) and a basis of $\pre^* (\Ii) = \clo{\pre^* (\Ii)}$ for a given ideal $\Ii$
can be computed.
 
We are interested in the exact number of steps which we need to reach an ideal. Thus, $\pre( \Ii')$ should be an ideal and we cannot use the technique from \cite{FinkelS01} for WSTSs.
We need to restrict our setting to \emph{strongly} WSTSs like in \cite{AbdullaCJT96}. 
First, we state our main result for SWSTSs, the decidability of resilience.
\begin{thm}[decidability of resilience] \label{mainthm}The explicit and the bounded resilience problem both are decidable
for strongly well-structured transition systems.
\end{thm}
We prove this theorem by giving a respective algorithm. It exploits a modified version of the ideal reachability algorithm in~\cite{AbdullaCJT96} (cp.\ Lemma~\ref{Abdulla}). We check in every iteration step inclusion in $\Ii^k = \bigcup_{0 \le j \le k} \pre^j (\Ii)$. Before doing so, we need a finite representation of $\post^*(s) \cap \Ji$ to check the inclusion in an ideal $\Ii'$. The next lemma uses that $\Ji$ and $\Ii'$ are downward- and upward-closed, respectively.

\begin{lem}[intersection with anti-ideal]\label{intersection} Let $A \subseteq \St$ be a set, $\Ji \subseteq \St$ an anti-ideal and $\Ii' \subseteq \St$ an ideal. 
Then $A \cap \Ji \subseteq \Ii' \Leftrightarrow (\clo{A}) \cap \Ji \subseteq \Ii' $.
\end{lem} 

This lemma enables us to prove Thm.~\ref{mainthm} given above. 
We iteratively determine the minimal $k$ satisfying $\post^*(s) \cap \Ji \subseteq \Ii^k$
(or stop, if there does not exist such $k$).

\begin{proof}[Proof of Theorem \ref{mainthm}] Let $B_\post$ be a basis of $\clo{\post^*(s)}$, 
$B_0$ a basis of $\Ii$, and $\Ji$ a decidable anti-ideal. 
For every~$k\ge0$, $\Ii^k$ is an ideal due to strong compatibility.
By applying Lemma~\ref{intersection} twice, we obtain 
\begin{equation*}
\post^*(s) \cap \Ji \subseteq \Ii^k \hspace*{2mm} \Leftrightarrow 
\hspace*{2mm} B_\post \cap \Ji \subseteq \Ii^k 
\end{equation*} 
for any $k\ge 0$. 
Since $B_\post$ is finite and $\Ji$ is a decidable anti-ideal, we can directly compute $B_\post \cap \Ji$. 
We perform a modification of the ideal reachability algorithm: Iteratively check whether \mbox{$B_\post \cap \Ji \subseteq \Ii^k$}. If this is the case, return $k_\minm=k$. 
Otherwise check whether $\Ii^{k+1}= \Ii^k$. If so, return $-1$ ($ \false $), otherwise continue.  
We have to make sure that every iteration step is decidable. In fact, we can compute
a basis of $\Ii^{k+1}$ if we have a basis of $\Ii^k$. This follows by the proof of Lemma~\ref{Abdulla}.
The stop condition is decidable and by Fact~\ref{stopcond} also sufficient. Soundness and completeness follow by the previous 
considerations and the fact~that 
\begin{equation*}
\post^*(s) \cap \Ji \subseteq \Ii^k \hspace*{2mm} \Leftrightarrow 
\hspace*{2mm}( \forall s'\in\Ji : (s \pfeil^* s') \Rightarrow \exists s''\in \Ii: s' \pfeil^{\le k} s''  )
\end{equation*}  
for any $k\ge 0$.
Termination is guaranteed by Lemma~\ref{Noether}.
 
To sum up, our algorithm 
decides whether there exists a $ k\geq 0 $ s.t. $\post^*(s) \cap \Ji \subseteq \Ii^k$,
and returns the minimal such $ k $ in the positive case.
Thus, it decides the bounded resilience problem. 
Given any $k$, we can check whether $k_\minm \le k$ and 
therefore decide the explicit resilience problem. 
\end{proof}

We denote the above described algorithm deciding resilience by \textsc{MinimalStep}($B_\post , \Ji, B_0$) and the used procedure returning a basis of $\pre(\clo{B'})$ by \textsc{PreBasis}($B'$). 
It is shown in \cite{Koenig17}, that such a prebasis is computable for GTSs,
and described in detail in \cite{Stueckrath16}.
The method \textsc{Min}$(B')$ minimizes a finite set $B'$
by deleting every element in $ B' $ for which there is already a smaller element in $ B' $.
\begin{algorithm}[!htb]\caption{Minimal $k$ Algorithm} \label{minstep2}
	\begin{algorithmic}[1]
		\Procedure{MinimalStep}{$B_\post , \Ji, B_0$}
		\Comment{$k_\minm$ (minimal upper bound for recovery time)/$-1$}
		\State $B \gets B_\post \cap \Ji $		
		\Comment{compute $B$ by taking out elements which are not in $\Ji$}
		\State $k \gets 0$ \Comment{increasing counter}
		\State $B_1 \gets B_0$  \Comment{basis of the current $\Ii^k$; $B_0$ is a given basis of $\Ii$}
		\State $B_2 \gets \emptyset$ \Comment{basis of the current $\Ii^{k+1}$}
		\While{true}
		\If{$B \subseteq \clo{B_1}$}
		\State \textbf{return} $ k $ \Comment{we found $k_\minm$}
		\Else 
		\State $B_2 \gets B_0 \cup  \textsc{PreBasis}(B_1 )$			\Comment{\textsc{PreBasis}$(B_1)$ computes the basis of $\clo{B_1}$}
		\State $B_2 \gets \textsc{Min}(B_2 )$         \Comment{\textsc{Min}$(B_2 )$ minimizes the set $B_2$}
		\If{$B_2 \subseteq \clo{B_1}$}
		\State \textbf{return} $ -1 $ \Comment{there exists no such $k$}
		\Else
		\State $B_1 \gets B_2$ \Comment{continue}
		\State $k \gets k +1$	
		\EndIf
		\EndIf 
		\EndWhile
		\EndProcedure \end{algorithmic}
\end{algorithm} 
 
In the proof of Thm.~\ref{mainthm}, it was crucial that we have \emph{strong} compatibility. 
This approach does not work for WSTSs in general. We loose precision when we only demand compatibility.
Thus, we conjecture that both resilience problems are undecidable for WSTSs in general, but this question remains still open.

%% file: sources/application.tex
\section{Application to Graph Transformation Systems}\label{sec:appl}

We apply the abstract results of the previous section to (joint) graph transformation systems and present a framework for verifying resilience of GTSs.
We exemplarily show how Petri nets fit in this setting and give also an example beyond Petri nets.

We considered ideals as safety, and decidable anti-ideals as ``bad'' conditions. In the setting of well-structured GTSs w.r.t.\ the subgraph order, these can be expressed as positive and negative constraints. Recall that, for a fixed class $\GG$ of graphs, $\Ii_c = \{G \in \GG \,\vert\, G \models c\}$ for a positive constraint $c$, and \mbox{$\Ji_{c'} = \{G \in \GG \,\vert\, G \models c' \}$} for a negative constraint $c'$.

\begin{fact}[ideals of graphs] Let $\GG_\ell$ be a class of graphs of bounded path length. Let $\Ii,\Ji \subseteq \GG_\ell$ be sets. 
\begin{enumerate}[(i)] 
	\item $\Ii $ is an ideal $\Leftrightarrow$ $\Ii=\Ii_c$ for a positive constraint $c$.
 	\item $\Ji $ is a decidable anti-ideal $\Leftrightarrow$ $\Ji=\Ji_c$ for a negative constraint $c$. 
 \end{enumerate} 
\end{fact}
Thus, for GTSs, our safety conditions are equivalent to positive constraints and bad conditions are equivalent to negative constraints. 
\begin{rem} More general graph constraints, e.g.,\ 
		$\forall (  
		\begin{tikzpicture}[baseline=(1)]
			\node[circle,inner sep=0pt,minimum size=5pt,draw=black,fill=white] (1)  at (0,0) {}; 
			\node at (0,0.2) {\scriptsize $1$};
		\end{tikzpicture}, 
		\exists ( 
		\begin{tikzpicture}[baseline=(1)]
			\clip(-0.15,-0.1) rectangle (0.3,0.3);
			\node[circle,inner sep=0pt,minimum size=5pt,draw=black,fill=white] (1)  at (0,0) {}; 
			\draw[-to] (1) edge [loop, min distance=3mm, out=60, in=120, looseness=6](1);  
			\node at (0.2,0) {\scriptsize $1$};
		\end{tikzpicture} ) )$, 
do not constitute ideals w.r.t.\ the subgraph order. The relation 
		$\begin{tikzpicture}[baseline=(1)]
		\clip(-0.15,-0.1) rectangle (0.3,0.3);
			\node[circle,inner sep=0pt,minimum size=5pt,draw=black,fill=white] (1)  at (0,0) {}; 
			\draw[-to] (1) edge [loop, min distance=3mm, out=60, in=120, looseness=6](1);  
			\node at (0.2,0) {\scriptsize $1$};
		\end{tikzpicture} 
		\wqo 
		\begin{tikzpicture}[baseline=(1)]
		\clip(-0.15,-0.1) rectangle (0.6,0.3);
			\node[circle,inner sep=0pt,minimum size=5pt,draw=black,fill=white] (1)  at (0,0) {}; 
			\draw[-to] (1) edge [loop, min distance=3mm, out=60, in=120, looseness=6](1);  
			\node at (0.2,0) {\scriptsize $1$}; \node[circle,inner sep=0pt,minimum size=5pt,draw=black,fill=white] at (0.5,0) {}; 
		\end{tikzpicture}$ 
shows that upward-closedness is not guarenteed. 
In special cases, (nested) graph constraints \cite{Rensink04,HabelP09} may yield ideals, e.g., the ideal in the later discussed Ex.~\ref{pathgame} can be expressed as $\forall (\Lnest, \exists (\LtotoLnest))$.
However, we conjecture that a generalization to more arbitrary (nested) graph contraints is not possible.
\end{rem}
\subsection{Verifying Resilience of Graph Transformation Systems}
Using the sufficient conditions for strong well-structuredness of K{\"o}nig \& St{\"u}ckrath \cite{Koenig17}, we obtain the decidability of both resilience 
problems for a subclass of GTSs. We need to use the subgraph order as wqo. Thus, we have the restriction of bounded4 path length for the considered
graph class. Instead of considering GTSs, we consider graph transition systems, i.e.,\ we always explicate the state~set.
Thm.~\ref{mainthm} and the result in \cite{Koenig17} (see Lemma~\ref{Koenig}) imply our main result for GTSs: 
\begin{thm}[decidability of resilience for well-structured GTSs] \label{mainGTS}
The explicit and the bounded resilience problem are decidable for graph transition systems which are of bounded path length (and equipped with the subgraph order).
\end{thm} 

As joint GTSs are also GTSs, the same sufficient conditions for strong well-structuredness apply.
\begin{fact}[strongly well-structured joint GTSs] \label{WSJGTS} Every annotated joint graph transition system which is of bounded path length is strongly well-structured (equipped with the subgraph order).
\end{fact}

An immediate consequence of Thm.~\ref{mainGTS} and Fact~\ref{WSJGTS} is the following:
\begin{cor}[decidability of resilience for joint GTSs]
The explicit and the bounded resilience problem are decidable for annotated joint graph transition systems which are of bounded path length (and equipped with the subgraph order).
\end{cor}
Thus, we can apply the algorithm \textsc{MinimalStep} described in Sec.~\ref{sec:decidability} to verify resilience of annotated joint graph transition systems.
We consider an ideal $\Ii_c$ for a positive constraint~$c$ with a given basis~$B_c$. 
The anti-ideal (bi-ideal)
is given by $\Ji_\senv=\{ G' \in \GG \oplus Q \oplus \Mrk \,\vert\, G' \models \exists \senv \}$.
We assume that a start graph 
$G \in \GG \oplus \{q_0 \}\oplus \{ \top \}$ 
and a basis $B_G$ of $\clo{\post^*(G)}$ are given. 
The \textsc{PreBasis} procedure for the subgraph order needed in the algorithm is given by K{\"o}nig \& St{\"u}ckrath in \cite{Koenig17} 
(and more detailed in~\cite{Stueckrath16}).
\begin{figure}[!htb]
		\begin{center} 
		\tikzset{
		 state/.style={
		        rectangle,
		        rounded corners,
		        draw=black,
		        minimum height=2em,
		        inner sep=2pt,
		        text centered,
		        },
		}
\begin{tikzpicture}[->,>=stealth',scale=0.83, transform shape]
\node[state,fill=whitegray] (box1) at (0,0.4) { \begin{tabular}{c} \hspace*{2mm} \\ \textsc{MinimalStep}$(B_G, \Ji_\senv, B_c)$  \\ \hspace*{2mm} \end{tabular}};

\draw (-5,1) -- node[above]{$(\SysEnv)'$}(-2.35,1);  
\draw (-5,0.4) -- node[above]{basis $B_G$}(-2.35,0.4);
\draw (-5,-0.2) -- node[above]{basis $B_c$}(-2.35,-0.2); 

\draw (2.35,0.4) -- node[above]{$k_\minm$/$\false$}(5,0.4); 

\end{tikzpicture} 
\end{center}
\caption{Verifying resilience in the adverse conditions approach.}
\end{figure}
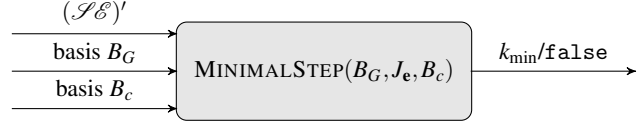

\subsection{Approximations}\label{sec:approximations}
We now drop an essential assumption for the decidability results in Sec.~\ref{sec:decidability}
by considering SWSTSs without a given basis of $\clo{\post^*(s)}$.
We show that we can still approximate $k_\minm$ from below (by $ k^\ell_\mathrm{un} $, $ \ell\in\mathbb{N} $) and above (by $ k_\mathrm{ov} $) by calculating corresponding approximations of (a basis of) $\clo{\post^*(s)}$.
The following function, called $\mu$-function, defines these approximations.
\begin{defn}[$\mu$-function, $ k^\ell_\mathrm{un} $, $ k_\mathrm{ov} $] 
	Let $\tuple{\St,\wqo,\pfeil}$ be a SWSTS, $\Ji \subseteq \St$ an anti-ideal, and $\Ii\subseteq \St$ an ideal. 
	We define the function $\mu : \Pot (\St) \to \mathbb{N} \cup \{\infty\}$ 
	as $\mu (A)= \minm(\{ k \in \mathbb{N} : A \cap \Ji \subseteq \bigcup_{j \le k}\pre^j (\Ii) \} \cup \{ \infty \})$ 
	where $\Pot(\St)$ is the power set of $\St$. 
	For $s \in \St$ and $\ell \in \mathbb{N}$, 
	let $k^\ell_\mathrm{un}:= \mu(\bigcup_{j \le \ell}\post^j (s))$ and $k_\mathrm{ov}:= \mu(\post^*(\clo{\{s\}}))$.
\end{defn}
Note that $k_\minm=\mu(\post^*(s))$ and that $k_{\minm}=\infty$ can be read as ``there is no such $k$''. 
By definition, $\mu$~is monotonic, i.e.,\ $A \subseteq B$ implies $\mu(A) \le \mu(B)$, and by Lemma~\ref{intersection}, $\mu(\clo{A})=\mu(A)$. 
For the under- and over-approximation, we consider a basis of $\clo{\bigcup_{j \le \ell}\post^j (s)}$ 
and a basis of $\clo{\post^*(\clo{\{s\}})}$, respectively. 
For every GTS of bounded path length, this under-approximation is feasible. 
We present an idea for performing the over-approximation by means of invertibility. 
\begin{fact}[weak invertibility]\label{inv}
	Let $\tuple{\GG, \Rightarrow_\R}$ be a graph transition system of bounded path length and $\R'$ a GTS 
	s.t.\ $G \Rightarrow^*_\R H$ iff $ H \Rightarrow^*_{\R'} G$ for all $G,H \in \GG$. 
	Then, for every $G \in \GG$, $\post_{\R}^* (\clo{\{G\}})=\pre_{\R'}^*(\clo{\{G\}})$ 
	and a basis of $\post_{\R}^* (\clo{\{G\}})$ is computable. 
\end{fact}
In particular, such an $ \R' $ exists
if $G \Rightarrow_r H$ iff $ H \Rightarrow_{r^{-1}} G$ for all $G,H \in \GG$, $r \in \R$,
where for a rule~$r=\tuple{L \pmor R}$ which is injective on its domain, $r^{-1}=\tuple{R \pmor L}$ is the \emph{inverse rule}. 
In general, $G \Rightarrow_r H$ only implies that there is a graph $G' \le G$ s.t.\ $H \Rightarrow_{r^{-1}} G'$, 
since an application of $r$ may have deleted dangling edges. 
However, in some classes of GTSs, e.g.,\ in Petri nets (see Sec.~\ref{sec:PN}), 
there are no dangling edges in both directions, and we can use the inverse rules for the over-approximation.
\begin{fact}[approximation] 
	Let $\tuple{\GG,\Rightarrow_\R}$ be a GTS of bounded path, 
	$\Ji \subseteq \GG$ an anti-ideal, $\Ii \subseteq \GG$ an ideal, and $G \in \GG$. 
	(i) For every $\ell\ge 0$, $k^\ell_{un}$ is computable and $k^\ell_{un} \le k_\minm$. 
	The sequence $ \tuple{k^\ell_{un}}_\ell $ converges to~$k_\minm$, eventually stabilizing. 
	(ii) Under the assumptions of Fact~\ref{inv}, $k_\mathrm{ov} $ is computable and $k_\mathrm{ov}\ge k_\minm$.  
\end{fact}
Note that $k^\ell_\mathrm{un}=\infty$ implies $k_\minm=\infty$, and $k_\mathrm{ov}<\infty$ implies $k_\minm<\infty$. Only if $k^\ell_\mathrm{un}=0$ and $k_\mathrm{ov}=\infty$, we gain no information about $k_\minm$.
The approximation results described above are visualized in Fig.~\ref{fig:under-and-over-approximation}.
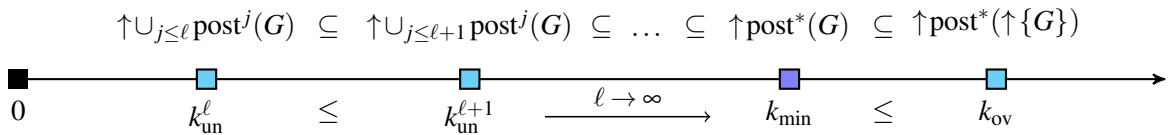
\begin{figure}[!htb] 
	\centering 
	\begin{tikzpicture}
	\tikzset{P/.style={rectangle,minimum size=0.01mm,fill=black}} 
	\draw[-stealth', thick] (0,0) node [P,label=below:{$0$}]{} 
		--(2.5,0) node [P,draw=black,fill=cyan!50,label=below:{$k^\ell_\mathrm{un}$}]{} 
		--(4.12,0) node [label=below:{$\le$}]{} 
		-- (6,0) node [P,draw=black,fill=cyan!50,label=below:{$k^{\ell+1}_\mathrm{un}$}]{}
		-- (7.75,0) node  {}
		-- (10.25,0) node [P,draw=black,fill=blue!50,label=below:{$k_\minm$}]{} 
		--(11.5,0) node [label=below:{$\le$}]{} 
		-- (13,0) node [P,draw=black,fill=cyan!50,label=below:{$k_\mathrm{ov}$}]{}  
		-- (15.25,0) {};
		\node at (2.5,0.7) 		{$\clo{\cup_{j \le \ell}\,\post^j (G)}$};
		\node at (6,0.7) 		{$\clo{\cup_{j \le \ell+1}\,\post^j (G)}$}; 
		\node at (4.12,0.7) 	{$\subseteq$}; 
		\node at (7.75,0.7)		{$\subseteq$} ;	
		\node at (8.35,0.6) 	{$\ldots$};
		\node at (9,0.7)		{$\subseteq$} ;
		\node at (10.25,0.7) 	{$\clo{\post^* (G)}$};
		\node at (11.5,0.7)		{$\subseteq$} ; 
		\node at (13, 0.75) 	{$\clo{\post^*(\clo{\{G\}})}$};
		\draw[-to,semithick] (7.0,-0.5)  --(9.2,-0.5) node[midway,above]{\small $\ell \to \infty$};
	\end{tikzpicture}
	\caption{Under- and over-approximation of $k_\minm$ by corresponding approximation of $\clo{\post^* (G)}$.}
	\label{fig:under-and-over-approximation}
\end{figure}

\subsection{An Example Class: Petri Nets}\label{sec:PN}
Petri nets \cite{Reisig85} are a common model for discrete distributed systems in computer science, often applied, e.g.,\ in logistics or supply chains \cite{Zhang09}. 
It is a classical example for strongly well-structured (graph) transition systems. We will give a definition of Petri nets and show how our example fits in this setting.

\begin{defn}[Petri nets]
	A \emph{Petri net} is a tuple $ \PN=\tuple{\pl,\tr,\fl} $
	with disjoint finite sets of \emph{places} $ \pl $ and \emph{transitions} $ \tr $,
	and a \emph{flow function} $ \fl: (\pl\!\times\!\tr)\cup (\pl\!\times\!\tr)\to \mathbb{N} $.
	A \emph{marking} in $ \PN $ is a multi-set $ \marking:\pl\to\mathbb{N} $ 
	that indicates the number of tokens on each place.
	$ \fl(x,y)=n>0 $ means there is an \emph{arc} of \emph{weight} $ n $
	from node $ x $ to $ y $ describing the flow of tokens in the net.
	A transition $ t\in\tr $ is \emph{enabled} in a marking $ \marking $
	if $ \forall p\in\pl: \fl(p,t)\leq\marking(p) $.
	If $ t $ is enabled, then $ t $ can \emph{fire} in $ \marking $,
	leading to a new marking~$ \marking' $ calculated by
	$\forall p\in\pl: \marking'(p)=\marking(p)-\fl(p,t)+\fl(t,p) $.
	This is denoted by $ \marking[t\rangle \marking' $.
	Usually, a Petri net $ \PN $ is equipped with an \emph{initial marking} $ \marking_0 $.
	The tuple $ \tuple{\PN,\marking_0} $ is then called a \emph{marked Petri net}.
\end{defn}

Any Petri net $ \PN $ can be interpreted as a transition system with
the states $ \St $ given by $ \markings(\PN) $, the set of all markings of $ \PN $,
and the transitions $ \pfeil $ given by $ \marking\pfeil\marking' \Leftrightarrow \exists t\in\tr:\marking[t\rangle \marking' $.
Together with the wqo~$ \leq_\text{PN} $, given by 
$ \forall \marking,\marking'\in \markings(\PN):
\marking\leq_\text{PN} \marking':\Leftrightarrow \forall p \in\pl:\marking(p)\leq \marking'(p)  $,
this constitutes a SWSTS.
For Petri nets, reachability and equivalent problems are decidable \cite{Reisig85,EsparzaN94}.
From this fact and the results in \cite{ValkJ85}, one can show that for Petri nets a basis of $ \clo{\post^*(\marking_0)} $ is computable:
In~\cite{ValkJ85}, it is shown that for any ideal $ \Ii $ of markings in a Petri net,
a basis of $ \Ii $ is computable iff for every \mbox{$ \omega $-marking}~$ \marking $ it is decidable whether
$ \Ii\cap\dlo{\{\marking\}}=\emptyset $. 
An $ \omega $-marking is a function $ \marking:\pl\to\mathbb{N}\cup\{\omega\} $,
and analogously to~before, $\dlo{\{\marking\}}:=\{ \marking ' \in \markings(\PN) \,\vert\, \forall p\in\pl: \marking'(p)\leq\marking(p)\lor \marking(p)=\omega \}$.
Since $ \clo{\post^*(\marking_0)} $ is an ideal, we can apply this result and ask whether 
$ \clo{\post^*(\marking_0)}\cap\dlo{\{\marking\}}=\emptyset $ is decidable.
This is obviously equivalent to $ \clo{\post^*(\marking_0)}\cap\dlo{\{\marking\}}\subseteq\emptyset $,
allowing us to apply Lemma~\ref{intersection}, since $ \emptyset $ is an ideal.
Thus, we now ask whether $$ \post^*(\marking_0)\cap\dlo{\{\marking\}}=\emptyset. $$
This problem corresponds 
to the so-called \emph{submarking reachability problem},
which is decidable~(cp.,\ e.g., \cite{EsparzaN94}), since it is recursively equivalent the to reachability problem.
Therefore, we get that a basis of~$ \clo{\post^*(\marking_0)} $ is computable. 

Petri nets can also be seen as an instance of GTSs, as shown in \cite{Baldan10}.
From that point of view, every transition corresponds to a graph transformation rule.
A marking is given by the structure of the Petri net represented as a graph,
with the number of tokens on a place represented by extra nodes connected to it, as in Fig.~\ref{fig:ruleApplication}.
The wqo $ \leq_\text{PN} $ then directly corresponds to the subgraph order.
Together with the start graph representing the initial marking,
interpreting the GTS as a WSTS results in exactly the same SWSTS above.
This means we can apply the algorithm deciding resilience in GTS to Petri nets.
We demonstrate this by the following example,
where we consider a Petri net that, when interpreted as a GTS, 
is exactly the supply chain modeled in Ex.~\ref{prodch}.

\begin{exmp}[supply chain as Petri net] \label{prodchPN}
	We consider a marked Petri net modeling a simplified scenario of a supply chain, shown in Fig.~\ref{fig:PetriNetExample}.
	As usual we depict places as circles, transitions as rectangles, and the flow as weighted directed arcs between them.
	In the example, all weights are $ 1 $ and therefore not indicated.
	Dots on places indicate the number of tokens on the respective place in the initial marking.
	\begin{figure}[!htb]
		\centering
		\begin{tikzpicture}
		\renewcommand{\xdis}{1cm}
		\renewcommand{\ydis}{0.85cm}
			\node[transition] (produce) [label={below:{\footnotesize$ \mathit{produce}$}}] {};
			\node[place,right=of produce] (product) [label={[label distance=-3pt]above:{\footnotesize$ \mathit{product}$}}]{};
			\node[transition,right=of product] (transport) [label={below:{\footnotesize$ \mathit{transport}$}}]{};
			\node[place,right=of transport,tokens=1] (warehouse)  [label={[label distance=-2pt]right:{\footnotesize$ \mathit{warehouse}$}}] {};
			\node[envtransition, above=of warehouse] (accident) [label={[label distance=-1pt]left:{\footnotesize$ \mathit{accident}$}}] {};
			\node[transition, above right=of warehouse] (ship1) [label={[label distance=-3pt]above:{\footnotesize$ \mathit{ship}_1$}}] {};
			\node[transition, below right=of warehouse] (ship2) [label={[label distance=-1pt]below:{\footnotesize$ \mathit{ship}_2$}}] {};
			\node[place, right=of ship1,tokens=1] (store1) [label={[label distance=-2.5pt]below:{\footnotesize$ \mathit{store}_1$}}]{};
			\node[place, right=of ship2,tokens=1] (store2) [label={[label distance=-3.5pt]above:{\footnotesize$ \mathit{store}_2$}}] {};
			\node[envtransition, right=of store1] (buy1) [label={[label distance=-3.5pt]above:{\footnotesize$ \mathit{buy}_1$}}]{};
			\node[envtransition, right=of store2] (buy2) [label={[label distance=-2.5pt]below:{\footnotesize$ \mathit{buy}_2$}}] {};
			
			\draw[->]
			(produce)		edge[post]	(product)
			(transport)		edge[pre]	(product)
							edge[post]	(warehouse)
			(accident)		edge[pre]	(warehouse)
			(ship1)	edge[pre]	(warehouse)
							edge[post]	(store1)
			(ship2)	edge[pre]	(warehouse)
							edge[post]	(store2)
			(buy1)			edge[pre]	(store1)
			(buy2)			edge[pre]	(store2)
			;
		\end{tikzpicture}
		\hspace{1cm}
		\begin{tikzpicture}
		\renewcommand{\xdis}{3cm}
		\renewcommand{\ydis}{0.7cm}
		\node[graphnode,fill=white,scale=2,inner sep=1pt] (init) {$ e $};
		\node[graphnode, right=of init] (d) {};
		\node[graphnode, right=of d] (dd) {};
		\node[graphnode,above=of d] (pt) {};
		\node[graphnode] at ($ (init)!0.5!(pt) $) (p) {};
		\node[graphnode,above=of dd] (ptpt) {};
		\node[graphnode] at ($ (pt)!0.5!(ptpt) $) (ptp) {};
		\node[graphnode] at ($ (d)!0.5!(dd) $)  (dp) [yshift=-2*\ydis] {};
		
		\node[left=of init, xshift=0.8*\xdis] (root) {};
		\draw[->]
		(root)	edge	(init)
		(init)	edge node[above,sloped,inner sep=1pt]{\footnotesize$\mathit{produce}$} (p)
		(p)	edge node[above,sloped,inner sep=1pt]{\footnotesize$\mathit{transport}$} (pt)
		(pt)	edge node[above,sloped,inner sep=1pt]{\footnotesize$\mathit{produce}$} (ptp)
		(ptp)	edge node[above,sloped,inner sep=1pt]{\footnotesize$\mathit{transport}$} (ptpt)
		(d)	edge node[below,sloped,inner sep=1pt]{\footnotesize$\mathit{produce}$} (dp)
		(dp)	edge node[below,sloped,inner sep=1pt]{\footnotesize$\mathit{transport}$} (dd)
		(init)	edge node[below,sloped,inner sep=1pt]{\footnotesize$\mathit{ship}_i$} (d)
		(d)	edge node[below,sloped,inner sep=1pt]{\footnotesize$\mathit{ship}_i$} (dd)
		(pt)	edge node[below,sloped,inner sep=1pt,pos=0.25]{\footnotesize$\mathit{ship}_i$} (dd)
		;
		\draw[->, rounded corners]
		(ptpt) -- ++(0,0.5cm) -- node[above,inner sep=1pt]{\footnotesize$\mathit{accident}$, $ \mathit{buy}_i $} ++(-2*\xdis,0) -- (init)
		;
		\draw[->, rounded corners]
		(dd) -- ++(0,-1cm) -- node[below,inner sep=1pt]{\footnotesize $ \mathit{buy}_i $} ++(-2*\xdis,0) -- (init)
		;
		\end{tikzpicture}
		\caption{A Petri net modeling a supply chain, and its control automaton.}
		\label{fig:PetriNetExample}
	\end{figure}
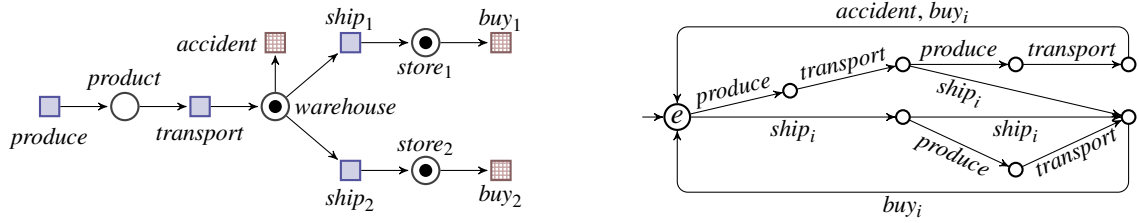

	The Petri net corresponds directly to the graph transformation rules in Ex.~\ref{prodch},
	with the blue transitions simulating $ \Sys $, and the red (checkered) transitions simulating $ \Env $.
	The initial marking represents the start graph. 
	Correspondingly, the control automaton has the same structure as in Ex.~\ref{prodch},
	with transitions replacing rules.
	Let $ \Ii=\{ \tuple{\marking,q} \,|\, \marking(\mathit{warehouse}),\marking(\mathit{store}_1),\marking(\mathit{store}_2)\geq 1 \land q\in Q \} $,
	i.e.,\ in the warehouse and in both stores products are available for shipping or purchase, respectively.
	The transitions corresponding to $ \Env $ reduce the number of tokens in the net.
	We consider the resilience problem with adverse conditions.
	By definition of the control automaton, 
	we know that $ \Ji_\senv =\{ \tuple{\marking,e} \vert \marking\text{ is a marking}\} $. 
	
We interpreted Ex. \ref{prodch}/ Ex. \ref{prodchPN} as joint GTS and 
applied a prototype implementation of the algorithm \textsc{MinimalStep} from Sec.~\ref{sec:decidability} to it. We obtained that $ k_\minm=6 $ is the smallest $ k $ for which the system is $ k $-step resilient: The following set $B^{\marking_0}_\senv$ is the intersection of a basis of $\clo{\post^* (\marking_0 )}$ with $\Ji_\senv$ where $M_0 = \vek{0,1,1,1,q_0}$. The first coordinate corresponds to (the number of tokens in) $ P $/$ \mathit{product} $, the second coordinate to $ W $/$ \mathit{warehouse} $, and the third and fourth coordinate correspond to $ S_1 $/$ \mathit{store}_1 $ and $ S_2 $/$ \mathit{store}_2 $, respectively.
\begin{align*} B^{\marking_0}_\senv = \{ &\vek{0,1,1,1}, \vek{0,5,0,0}, \vek{0,0,3,0} , \vek{0,0,0,3}, \vek{0,1,2,0},  \\ & \vek{0,1,0,2}, \vek{0,0,2,1}, \vek{0,0,1,2}, \vek{0,3,1,0}, \vek{0,3,0,1}\} \times \{q_0 \}  \end{align*}
We computed $B^k$, a basis of $\Ii^k$, for $1 \le k \le 21$. We only give $B^k \cap \Ji_\senv$ for $1 \le k \le 6$: 
\begin{align*} 
	B^1 \cap\Ji_\senv &= \{ \vek{0,1,1,1}, \vek{0,2,0,1}, \vek{0,2,1,0}\} \times \{ q_0 \} \\
	B^2 \cap \Ji_\senv &= \{ \vek{0,0,1,1},  \vek{0,2,0,1}, \vek{0,2,1,0}, \vek{0,3,0,0}\}  \times \{ q_0 \} \\
	B^3 \cap \Ji_\senv &= \{ \vek{0,0,1,1},  \vek{0,1,0,1}, \vek{0,1,1,0}, \vek{0,3,0,0}\}  \times \{ q_0 \} \\
	B^4 \cap\Ji_\senv &= B^5 \cap \Ji_\senv = B^3 \cap \Ji_\senv \\
	B^6 \cap \Ji_\senv &=\{ \vek{0,0,1,1},  \vek{0,1,0,1}, \vek{0,1,1,0}, \vek{0,3,0,0}, \vek{0,0,2,0},\vek{0,0,0,2} \}  \times \{ q_0 \}
\end{align*}
We obtain $B^{\marking_0}_\senv \not \subseteq \clo{B^k \cap \Ji_\senv}$ for $1 \le k \le 5$, 
but $B^{\marking_0}_\senv \subseteq \clo{B^6 \cap \Ji_\senv}$. 
Thus, $k_\minm =6$. 
\end{exmp}

\subsection{An Example beyond Petri Nets}
We give an example for a joint GTS which cannot be modeled by a (finite) Petri net and verify its resilience.
\begin{exmp}[path game] \label{pathgame}
	Consider two fixed locations represented by nodes labeled with $L$. 
	Points between them are represented by black nodes. 
	The system tries to construct two directed paths of length~$ 2 $ between the locations, one path forth and one back, using the rules $ \Sys $ in Fig.~\ref{fig:pathgame}.
	The respective ideal is therefore given by $\exists (\BEIdi) \lor \exists ( \BEIdii)$. 
	The environment deletes edges in the graph, corresponding to $ \Env $ in Fig.~\ref{fig:pathgame}. 
	The control automaton is alternating: 
	\begin{equation*}
	\begin{tikzpicture}[baseline=(S.center),scale=0.75,transform shape]
	\renewcommand{\xdis}{2cm}
	\node[circle,draw=black] (S) {$e $};
	\node[circle,draw=black,right=of S] (E) {$ s $};
	\node[above=of S,xshift=-2mm,yshift=-3mm] (A) {$ A $};
	\path[->]
	(S) edge[bend left = 20] node[above] {$ \Sys $} (E)
	(E) edge[bend left = 20] node[below] {$ \Env  $} (S)
	([xshift=-9mm,yshift=-0mm]S.center) edge (S)
	;
	\end{tikzpicture}
	\end{equation*}	
	Thus, one may consider this as a game with alternating turn order. 
	The system can 
	(i) create a new middle point connected to the locations by the rule \texttt{New}, 
	(ii) create two parallel edges provided that there is one by the \texttt{Para}-rules, 
	(iii) reverse the direction of an edge by the \texttt{Rev}-rules, and 
	(iv) merge two middle points each of which are connected to a different location by the \texttt{Mer}-rules.
	We ask whether the system can construct the two directed paths of length~$ 2 $ in a bounded number of rounds (steps) when the environment made its turn, regardless of the current situation. 
	If so, what is the minimal number of~steps?
	
	We can reach the graph $G_{LL}:=\LaLG$ (modulo isolated nodes) when the system is only changing the direction of edges. Hence, $\tuple{ G_{LL}, e} \in \clo{\post^*(G)} \cap \Ji_\senv$ for any start graph $G$ with exactly two locations, arbitrarily many middle points, and arbitrary edges between middle points and locations. Therefore, we only check when $\tuple{ G_{LL}, e}$ occurs the first time in a basis $B^k$. We applied a prototype implementation of the algorithm in Sec.~\ref{sec:decidability} to this example and obtained $B^{13} \cap \Ji_\senv = \{ \tuple{G_{LL},e}\}\not\subseteq B^{12}$ by computation of 
		\begin{align*} 
		B^{12} =  \hspace*{1.2mm}\big(\big\{& \BEi, \hspace*{3mm}\BEii,\hspace*{3mm} \BEiii,\hspace*{3mm} \BEiv,\hspace*{3mm} \BEv ,\hspace*{3mm}\BEvi, \\ 
				&\BEvii, \hspace*{3mm} \BEix ,\hspace*{3mm} \BEviii , \hspace*{3mm}\BEx, \hspace*{3mm}\BExi, \hspace*{3mm}\BExii \big\} \times \{ s\}\big)  \\  
				\cup \hspace*{1.2mm}\big(\big\{ &\BExiii, \hspace*{3mm} \BExiv  \big\} \times \{e \}\big). 
		\end{align*} 
		Thus, $k_\minm = 13$.
\end{exmp}
\begin{figure}[!tb]
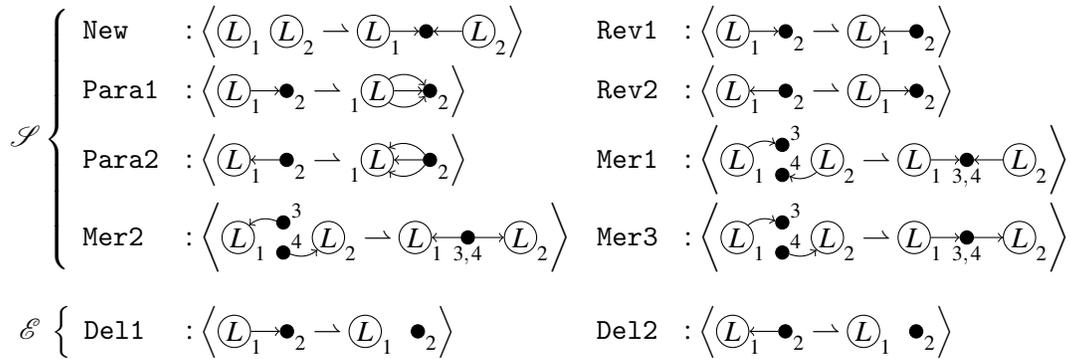

	\begin{align*}
				\Sys &\left\{ 
				\begin{array}{llll} 
				\texttt{New} &: \left\tuple{ 
					\LaL
					\rightharpoonup 
					\LtoMidotL
					\right} & \texttt{Rev1}&: \left\tuple{ \LtoMid \rightharpoonup \LotMid \right} \vspace*{1mm}\\ 
				\texttt{Para1} &: \left\tuple{ 
					\LtoMid
					\rightharpoonup 
					\LtoMidPara
					\right}  & \texttt{Rev2}&: \left\tuple{ \LotMid \rightharpoonup \LtoMid \right} \vspace*{1mm}\\ 
				\texttt{Para2} &: \left\tuple{ 
					\LotMid
					\rightharpoonup 
					\LotMidPara
					\right} &  \texttt{Mer1} &: \left\tuple{ 
					\MergeLiii
					\rightharpoonup 
					\MergeRi
					\right}  \vspace*{1mm} \\
			\texttt{Mer2} &: \left\tuple{\MergeLi \rightharpoonup \MergeRii \right} &
	\texttt{Mer3} &: \left\tuple{\MergeLii \rightharpoonup \MergeRiii \right}
				\end{array} 
				\right. \\[3mm]
				\Env &\hspace*{1mm}\left\{ 
				\begin{array}{llll}
				\texttt{Del1} \hspace*{0.19cm}&: \left\tuple{  
					\LtoMid
					\rightharpoonup 
					\LaMid
					\right} &\hspace*{1.53cm} \texttt{Del2}&: \left\tuple{ 	\LotMid	\rightharpoonup \LaMid \right}
				\end{array} 
				\right.
	\end{align*}
			\caption{Components of the joint GTS}\label{fig:pathgame}
		\end{figure} 
Note that we consider equivalence classes of graphs modulo isolated middle points. This has no effect on the well-structuredness of this example. Also note that leaving out the rules for merging has only a slight impact on the bases and no effect on $k_\minm$.
\subsection{Adverse Conditions vs. Error States}
We compare the adverse conditions approach with the error state approach. 
As pointed out, these two views of resilience are not equivalent. 
While every system that is resilient w.r.t.\ error states (i.e.,\ $ \Ji=\St\setminus\Ii $)
is also resilient w.r.t.\ adverse conditions (i.e., $ \Ji=\Ji_\senv $)
due to $ \Ji_\senv\setminus\Ii \subseteq \St\setminus\Ii $ 
(meaning that if we can reach $ \Ii $ from \emph{every} state, then also from every state in $ \Ji_\senv $),
the opposite does not hold in general.

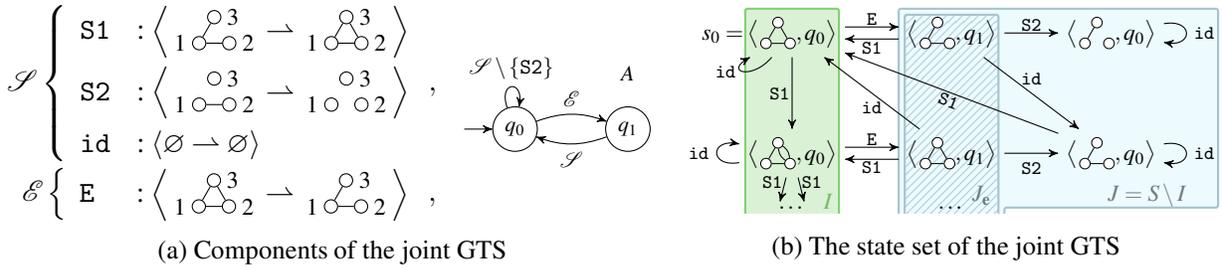
\begin{figure}[!htb]
		\centering
		\begin{subfigure}[h]{0.54\textwidth}
			\centering
			\begin{minipage}{0.65\textwidth}
				$\Sys \left\{ 
				\begin{array}{ll} 
				\texttt{S1} &: \left\tuple{ 
					\trianglegraph{1}{0}{1}
					\rightharpoonup 
					\trianglegraph{1}{1}{1}
					\right} \vspace*{1mm}\\ 
				\texttt{S2} &: \left\tuple{ 
					\trianglegraph{1}{0}{0}
					\rightharpoonup 
					\trianglegraph{0}{0}{0}
					\right}  \vspace*{1mm}\\ 
				\texttt{id} &: \left\tuple{ 
					\emptyset
					\rightharpoonup 
					\emptyset
					\right} 
				\end{array} 
				\right.,$ \vspace*{0mm}\\
				\hspace*{2mm}$\Env \left\{ 
				\begin{array}{ll}
				\texttt{E}\phantom{\mathtt{A}} &: \left\tuple{  
					\trianglegraph{1}{1}{1}
					\rightharpoonup 
					\trianglegraph{1}{0}{1}
					\right}
				\end{array} 
				\right.,$ 
			\end{minipage}\hfill
			\begin{minipage}{0.31\textwidth}
				\centering
			\begin{tikzpicture}[baseline=(S.center),scale=0.75,transform shape]
			\renewcommand{\xdis}{2cm}
			\node[circle,draw=black] (S) {$ q_0 $};
			\node[circle,draw=black,right=of S] (E) {$ q_1 $};
			\node[above=of E] (A) {$ A $};
			\path[->]
			(S) edge[loop, out=110, in=80,min distance=5mm] node[above,xshift=0mm] {$ \Sys\setminus\{\texttt{S2}\} $} (S)
			(S) edge[bend left = 20] node[above] {$ \Env $} (E)
			(E) edge[bend left = 20] node[below] {$ \Sys  $} (S)
			([xshift=-9mm,yshift=-0mm]S.center) edge (S)
			;
			\end{tikzpicture}
			\end{minipage}\hfill
			\caption{Components of the joint GTS}\label{fig:adverseJGTS}
		\end{subfigure} 
	\hfill
		\begin{subfigure}[h]{0.44\textwidth}
		\centering
		\begin{tikzpicture}[scale=0.8, transform shape]
		\renewcommand{\xdis}{2.65cm}
		\renewcommand{\ydis}{2cm}
			\node[] (init) {$ \big\langle\trianglegraphP{1}{1}{1}{0}, q_0 \big\rangle$};
			\node[right=of init] (E) {$ \big\langle\trianglegraphP{1}{0}{1}{0},q_1\big\rangle $};
			\node[right=of E] (ES2) {$ \big\langle\trianglegraphP{0}{0}{1}{0},q_0\big\rangle $};
			\node[below=of init] (S1) {$ \big\langle\trianglegraphP{1}{0}{1}{1},q_0\big\rangle $};
			\node[below=of E] (S1E) {$ \big\langle\trianglegraphP{1}{1}{1}{0},q_1\big\rangle $};
			\node[below=of ES2] (S1ES2) {$ \big\langle\trianglegraphP{1}{0}{1}{0},q_0\big\rangle $};
			\node at (init) [xshift=-11.5mm,yshift=-0.5mm] (initlabel) {$ s_0= $};
			\node at (S1) [yshift=-9mm] (dots) {$ \dots $};
			\node at (S1E) [yshift=-9mm] (dots2) {$ \dots $};
			\path[->]
			(init) edge[loop, out=-130,in=-150, min distance=5mm] node[left] {{\footnotesize $ \mathtt{id} $}} (init)
			(init) edge node[left, inner sep =1pt] {{\footnotesize $ \mathtt{S1} $}} (S1)
			([yshift=1mm]init.east) edge node[above, inner sep =1pt] {{\footnotesize $ \mathtt{E} $}} ([yshift=1mm]E.west)
			(S1ES2) edge node [below, inner sep =1pt,sloped] {\footnotesize $ \mathtt{S1} $} (init)
			([yshift=-1mm]E.west) edge node[below, inner sep =1pt] {{\footnotesize $ \mathtt{S1} $}} ([yshift=-1mm]init.east)
			(S1) edge[loop, out=-170,in=-190, min distance=5mm] node[left] {{\footnotesize $ \mathtt{id} $}} (S1)
			([yshift=1mm]S1.east) edge node[above, inner sep =1pt] {{\footnotesize $ \mathtt{E} $}} ([yshift=1mm]S1E.west)
			([yshift=-1mm]S1E.west) edge node[below, inner sep =1pt] {{\footnotesize $ \mathtt{S1} $}} ([yshift=-1mm]S1.east)
			(E) edge node[above, inner sep = 1pt] {\footnotesize $ \mathtt{S2} $} (ES2)
			(E) edge node[above, inner sep = 4pt] {\footnotesize $ \mathtt{id} $} (S1ES2)
			(S1E) edge node[below, inner sep=4pt] {\footnotesize $ \mathtt{id} $} (init)
			(S1E) edge node[below, inner sep=4pt] {\footnotesize $ \mathtt{S2} $} (S1ES2)
			(ES2) edge[loop, out=10,in=-10, min distance=5mm] node[right] {{\footnotesize $ \mathtt{id} $}} (ES2)
			(S1ES2) edge[loop, out=10,in=-10, min distance=5mm] node[right] {{\footnotesize $ \mathtt{id} $}} (S1ES2)
			
			([xshift=-1mm]S1.south) edge node[left,inner sep=1pt, pos=0.2] {{\footnotesize $ \mathtt{S1} $}} ([yshift=-4mm,xshift=-2mm]S1.south)
			([xshift=1mm]S1.south) edge node[right,inner sep=1pt, pos=0.2] {{\footnotesize $ \mathtt{S1} $}} ([yshift=-4mm,xshift=2mm]S1.south)
			;			
			
			\begin{pgfonlayer}{bg}
			\draw[thick,draw=mantis!100,fill=mantis!30,rounded corners=0.3mm]
			([xshift=1mm,yshift=-6mm]S1.south west) 
				-- ([xshift=1mm,yshift=0mm]S1.west|-init.north) 
				-- ([xshift=-1mm,yshift=0mm]init.north east)
				-- ([xshift=-1mm,yshift=-6mm]S1.south east) 
			;
			\node at (S1) [xshift=6mm, yshift=-8mm, mantis!90!black] (I) {$ \Ii $}; 
			
			\draw[thick,draw=lightcyan!80!black,fill=lightcyan!40,rounded corners=0.3mm]
			([xshift=0mm,yshift=-6mm]S1E.south west) 
			-- ([xshift=0mm,yshift=0mm]E.north west) 
			-- ([xshift=9mm,yshift=0mm]ES2.east|-E.north)
			-- ([xshift=9mm,yshift=-5mm]ES2.east|-S1E.south)
			-- ([xshift=0mm,yshift=-5mm]S1E.south east) 
			-- ([xshift=0mm,yshift=-6mm]S1E.south east)  
			;
			\node at (S1ES2) [xshift=6mm, yshift=-7mm, lightcyan!50!black] (J) {$ \Ji=\St\setminus \Ii $};
			
			\draw[thick,draw=lightcyan!80!black,pattern= north east lines, pattern color=lightcyan!85!black,rounded corners=0.3mm]
			([xshift=1mm,yshift=-6mm]S1E.south west) 
			-- ([xshift=1mm,yshift=-1mm]E.north west) 
			-- ([xshift=-1mm,yshift=-1mm]E.north east)
			-- ([xshift=-1mm,yshift=-6mm]S1E.south east) 
			;
			\node at (S1E) [xshift=5mm, yshift=-7mm, lightcyan!50!black] (JE) {$ \Ji_\senv $}; 
			\end{pgfonlayer}
		\end{tikzpicture}
		\caption{The state set of the joint GTS}\label{fig:adverseStatespace}
	\end{subfigure} 
		\caption{A joint GTS example that is $ 1 $-step resilient w.r.t.\ $ \Ji_\senv $ (adverse conditions), 
			but not resilient w.r.t.\ $ \Ji=\St\setminus\Ii $ (error states).}
		\label{fig:AdverseVsError}
\end{figure}
We do not define a restriction on the system/environment to allow more freedom of modeling
but our counterexample in Fig.~\ref{fig:AdverseVsError} captures the adverse effect of the environment.
The joint GTS in Fig.~\ref{fig:adverseJGTS},
together with a start graph $ \trianglegraphS{1}{1}{1}{0.7} $,
results in the state set in Fig.~\ref{fig:adverseStatespace}.
A basis of $ \Ii $ is given by~$ \langle\trianglegraphS{1}{1}{1}{0.7},q_0\rangle $.
From the definition of $ A $, we see that $ \Ji_\senv $ is given by 
$ \{ \tuple{G, q_1}\,|\, \trianglegraphS{0}{0}{0}{0.7}\leq G \} $,
indicated by the hatched area. 
We see that from every reachable state in $ \Ji_\senv $
we can reach $ \Ii $ in one step, which means that 
the system is $ 1 $-step resilient w.r.t.\ adverse conditions. 
On the other hand, we cannot reach $ \Ii $ from the state
$ \tuple{\trianglegraphS{0}{0}{1}{0.7},q_0}\in\St\setminus\Ii $, 
which is reachable from $ \Ji_\senv $ when the ``wrong'' system rule is applied.
This means the system is \emph{not} resilient w.r.t.\ error states.

If, due to the structure of a joint GTS, we can reach $\Ji_\senv$ from every reachable error state,
as, e.g., in Ex.~\ref{prodch}, 
both approaches coincide.
The computed $k_\minm$'s then only differ by at most the index $k(\Ji_\senv)$.

%% file: sources/relatedwork.tex
\label{sec:related}
We use SPO graph transformation for modeling systems as in  \textbf{L{\"o}we} \cite{Loewe91} (see also \textbf{Ehrig et al.\ }\cite{Ehrig97}).  

Our notion of joint GTSs is a special case of graph-transformational interacting systems. Another approach considering dependencies can be found, e.g.,\ in \textbf{Corradini et al.} \cite{CorradiniFR08}. 

The concept of resilience is broadly used in different areas, e.g.,\ in industrial control systems \cite{Trivedi09,Rieger13}, with varying definitions. Following these ideas, we formulated resilience in the abstract settings of TSs and GTSs. 
Our interpretation of resilience captures recovery in bounded time. 

\textbf{Abdulla et al.\ }\cite{AbdullaCJT96} show the decidability of 
ideal reachability (coverability), eventuality properties and simulation in (labeled) SWSTSs.
We use the presented algorithm to show the decidability of resilience problems in SWSTSs. 

\textbf{Finkel \& Schnoebelen} \cite{FinkelS01} show that the concept of well-structuredness is ubiquitous in computer science by providing a large class of example models (e.g.,\ Petri nets and their extensions, communicating finite state machines, lossy systems, basic process algebras). 
Moreover, they give several decidability results for systems with different degrees of well-structuredness. 
They also generalize the algorithm of~\cite{AbdullaCJT96} to (not necessarily strongly) WSTSs to show decidability of coverability. 

In \cite{Koenig17}, \textbf{K{\"o}nig \& St{\"u}ckrath} extensively study the well-structuredness of GTSs. 
More detailed considerations can be found in 
\cite{Stueckrath16}.
They identify three types of wqos (minor, subgraph, induced subgraph) on graphs based on results of \textbf{Ding}~\cite{Ding92} and \textbf{Robertson \& Seymour}~\cite{Robertson04}. The fact that the subgraph order is a wqo on graphs of bounded path length while the minor order allows all graphs comes with a trade-off:
For obtaining well-structuredness w.r.t.\ the minor order, the GTS must contain all edge contraction rules, i.e.,\ it must be a ``lossy'' GTS. 
On the other hand, all GTSs (without application conditions) are strongly well-structured on graphs of bounded path length w.r.t.\ the subgraph order. This result enables us to apply our abstract results to GTSs (in particular, we use the pred-basis procedure in the case of the subgraph order for our algorithm). 
In our setting, the regarded wqo is the subgraph order since it yields strong compatibility. They also generalize the notion of well-structured transition systems by
regarding $Q$-restricted WSTSs whose state sets needs not to be a wqo but rather a subset $Q$ of the states is a wqo. 
K{\"o}nig \& St{\"u}ckrath develope a backwards algorithm based on \cite{FinkelS01} for $Q$-restricted WSTSs obtaining decidability of coverability under additional assumptions. For SWSTSs, this approach coincides with the ideal reachability algorithm \cite{AbdullaCJT96}.

All in all, our result for SWSTSs uses a modification of Abdulla et.\ al \cite{AbdullaCJT96}, and our application to GTSs additionally uses the predecessor-basis procedure from K{\"o}nig \& St{\"u}ckrath \cite{Koenig17} in every computation step. It can also be seen as a modification of the backwards analysis of K{\"o}nig \& St{\"u}ckrath \cite{Koenig17} in the case of the subgraph order. We summarize the relations of our results and the used concepts in Fig. ~\ref{fig:relConcepts}. 
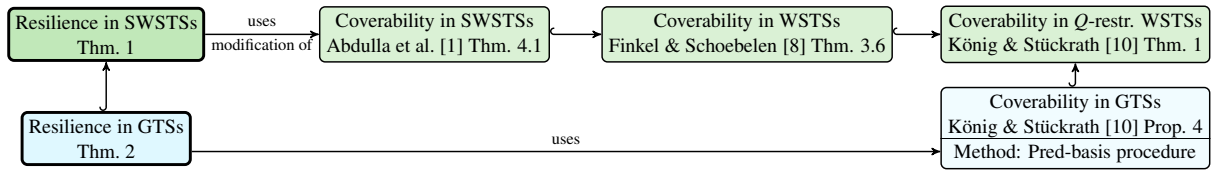
\begin{figure}[!htb]
	\centering
	\resizebox{\textwidth}{!}{%
		\begin{tikzpicture}
		\renewcommand{\ydis}{2cm}
		\renewcommand{\xdis}{6.3cm}
		\node[rectangle, rounded corners=3pt, draw=black,align=center,ultra thick,fill=mantis!50] (resilienceSWSTS) {Resilience in SWSTSs\\ Thm.~\ref{mainthm}};
		\node[rectangle, rounded corners=3pt, draw=black,align=center,right= of resilienceSWSTS,thick,fill=mantis!30] (coverabilitySWSTS) {Coverability in SWSTSs\\ Abdulla et al.\ \cite{AbdullaCJT96} Thm. 4.1};
		\node[rectangle, rounded corners=3pt, draw=black,align=center,right= of coverabilitySWSTS,thick,fill=mantis!30,xshift=-3mm] (coverabilityWSTS) {Coverability in WSTSs\\ Finkel \& Schoebelen \cite{FinkelS01} Thm. 3.6};
		\node[rectangle, rounded corners=3pt, draw=black,align=center,right= of coverabilityWSTS,thick,fill=mantis!30] (coverabilityQWSTS) {Coverability in $Q$-restr. WSTSs \\ K{\"o}nig \& St{\" u}ckrath  \cite{Koenig17} Thm. 1};
		\node[rectangle, rounded corners=3pt, draw=black,align=center,below= of resilienceSWSTS,ultra thick,fill=lightcyan!70] (resilienceGTS) {Resilience in GTSs\\ Thm.~\ref{mainGTS}};
		\node[rectangle, rounded corners=3pt, draw=black,align=center,below= of coverabilityQWSTS,thick,fill=lightcyan!30,yshift=2mm] (coverabilityGTS) {Coverability in GTSs \\ K{\"o}nig \& St{\" u}ckrath \cite{Koenig17} Prop. 4 \\ Method: Pred-basis procedure};
		
		\draw[->,thick]
		(resilienceSWSTS) edge node[align=center] {{\footnotesize uses}\\{\footnotesize modification of}} (coverabilitySWSTS)
		;
		\draw[right hook->,thick]
		(resilienceGTS) edge  (resilienceSWSTS) 
		(coverabilitySWSTS) edge (coverabilityWSTS)
		(coverabilityWSTS) edge (coverabilityQWSTS)
		(coverabilityGTS) edge (coverabilityQWSTS)
		;
		\draw[-]
		([yshift=-2mm]coverabilityGTS.west) edge ([yshift=-2mm]coverabilityGTS.east)
		;
		\draw[->, rounded corners=3pt,thick]
		([yshift=-2.5mm]resilienceGTS.east) -- node[above] {\footnotesize uses} ([yshift=-2.5mm]coverabilityGTS.west|-resilienceGTS.east)
		;
		\end{tikzpicture}}
	\caption{Relation of the decidability results for resilience (bold) and the results in related work.
	The bottom (blue) and the top (green) layers contain decidability results for GTSs and WSTSs, respectively. The hooked arrows ($\injto$) mean ``generalized to'' or ``instance of''.}
	\label{fig:relConcepts}
\end{figure}

%% file: sources/conclusion.tex
\label{sec:conc}
We provided a definition of resilience in an abstract framework, namely the explicit and the bounded resilience problem,
and proved decidability of both problems for strongly well-structured transition systems.
By application of this theory, we obtained decidability results for GTSs of bounded path length, 
and in particular a verification framework for GTSs which incorporates adverse conditions.

Our results require that a basis of the upward-closure of all successors is given. 
Although determining this basis for GTSs is a difficult task, it is computable for Petri nets and can be computed for other GTSs in special cases. 
We showed how to approximate such a basis when the assumption is dropped,
thereby approximating the answer to the resilience problems.
In this paper, the used well-quasi-order on graphs is the subgraph order. 
For the proof, the requirement of \emph{strong} compatibility is crucial. 
Our approach does not work for lossy GTSs which are well-structured w.r.t.\ the minor order. 
We conjecture that both resilience problems are undecidable for lossy GTSs. 
Ideals w.r.t.\ the subgraph order can be represented by positive basic graph constraints. 
In general, nested graph constraints do not constitute ideals.

\emph{Future work}. We will investigate on
(1)~the (un)decidability of resilience for WSTSs/lossy GTSs,
(2)~synthesis of resilient GTSs, i.e.,\ using the presented approach to construct provably resilient GTSs, and
(3)~the computability of a basis of the upward-closure of all successors for (a subclass of) strongly well-structured GTSs.
Regarding (2), we will investigate on the construction of strongly well-structured GTSs. 
Regarding (3), we will consider further methods for 
achieving approximation results for resilience.\\[2mm]
\noindent\textbf{Acknowledgment.} 
We are grateful to Annegret Habel, Christian Sandmann, and the anonymous reviewers for their helpful comments to this paper. 
We thank Barbara K\"{o}nig for the discussion about approximation and computation of the upward-closure of all successors, 
and Detlef Plump for the note on graph classes of bounded path length and bounded node degree.